\documentclass[journal]{IEEEtran}
 \usepackage{amsmath,amssymb}
 \usepackage{subfigure}
 \usepackage{graphicx,graphics,color,psfrag}
 \usepackage{cite,balance}
 \usepackage{caption}
 \allowdisplaybreaks
 \usepackage{algorithm}
 \usepackage{accents}
 \usepackage{amsthm}
 \usepackage{bm}
 \usepackage{algorithmic}
 \usepackage[english]{babel}
 \usepackage{multirow}
 \usepackage{enumerate}
 \usepackage{cases}
 \usepackage{stfloats}
 \usepackage{dsfont}
 \usepackage{color,soul}
 \usepackage{amsfonts}
 \usepackage{cite,graphicx,amsmath,amssymb}
 \usepackage{subfigure}
 \usepackage{fancyhdr}
 \usepackage{hhline}
 \usepackage{graphicx,graphics}
 \usepackage{array,color}
 \usepackage{amsmath}
 \usepackage{booktabs}

\newtheorem{theorem}{Theorem}

\newtheorem{lemma}{Lemma}
\newtheorem{corollary}{Corollary}

\newtheorem{proposition}{Proposition}

\newtheorem{remark}{\bf Remark}
\def\phi{\varphi}

\def\l{\left}
\def\r{\right}
\def\({\left(}
\def\){\right)}

\def\b0{{\mathbf{0}}}

\def\bP{{\mathbf{P}}}

\newcommand{\nn}{\nonumber}

\begin{document}

\title{\huge Cache-Enabled Heterogeneous Cellular Networks: \\
Optimal Tier-Level Content Placement}
\author{Juan Wen, Kaibin Huang, Sheng Yang and Victor O. K. Li  
\thanks{\noindent J. Wen, K. Huang, and Victor O. K. Li are  with the Deptment of Electrical and Electronic Engineering,  The University of  Hong Kong, Hong Kong, China (Email: jwen@eee.hku.hk, huangkb@eee.hku.hk, vli@eee.hku.hk).}
\thanks{\noindent S. Yang is with CentraleSupelec, Gif-sur-Yvette Cedex 91192, France (Email: 
sheng.yang@supelec.fr)}
\thanks{\noindent The work was supported by Hong Kong Research Grants Council under the PROCORE-France/Hong Kong Joint Research Scheme with the grant number F-HKU703/15T.}
}
\maketitle

\begin{abstract}
 
Caching popular contents at \emph{base stations} (BSs)  of a \emph{heterogeneous cellular
network} (HCN) avoids frequent information passage from content providers to the network
edge, thereby reducing   latency and alleviating traffic congestion in backhaul links. The
potential of caching at the network edge for tackling 5G challenges   has motivated the
recent studies of optimal  content placement in large-scale HCNs. However, due to the
complexity of network performance analysis, the existing strategies were designed mostly
based on approximation,  heuristics and intuition. In general, the optimal strategies for
content placement in HCNs remain largely unknown and deriving them forms the theme of this
paper. To this end, we adopt the popular random HCN model where  $K$ tiers of BSs are modeled
as independent \emph{Poisson point processes} (PPPs) distributed in the plane with different
densities. Further, the \emph{random caching} scheme is considered where each of a given set
of  $M$ files with corresponding popularity measures  is placed at each BS of a particular
tier with a corresponding  probability, called \emph{placement probability}. The
probabilities are identical for all BSs in the same tier but vary over tiers, giving the name
\emph{tier-level content placement}. We consider the network performance metric, \emph{hit
probability}, defined as the probability that a file requested by the typical user is
delivered successfully to the user. Leveraging existing results on HCN performance, we
maximize the hit probability over content placement probabilities, which yields the optimal tier-level placement policies. For the case of uniform  received signal-to-interference thresholds for successful transmissions for BSs in different tiers, the policy is in closed-form where the placement probability for a particular file is proportional to the \emph{square-root of the corresponding popularity measure} with an offset depending on BS caching capacities. For the general case of non-uniform SIR thresholds,  the optimization problem is non-convex and a sub-optimal placement policy is designed by approximation, which has  a similar structure as in the case of uniform SIR thresholds and shown by simulation to be close-to-optimal. 

\end{abstract}


\section{Introduction} 
The last decade has seen multimedia contents becoming dominant in mobile data traffic
\cite{Cisco2015-2020}. As a result, a vision for 5G wireless systems is to enable high-rate
and low-latency content delivery, e.g., ultra-high-definition video streaming
\cite{5G_Survey}. The key challenge for realizing this vision is that transporting large
volumes of data from content providers to end users causes severe traffic congestion in
backhaul links, resulting in rate loss and high latency \cite{5Gbackhual}.  On the other hand, the dramatic advancement of the  hard-disk technology makes it feasible  to deploy  large storage (several to dozens of TB) at  the network edge~(e.g., base stations (BSs) and dedicated access points) at low cost \cite{FemtoD2D_Caching}. In view of these, caching popular
contents at the network edge has emerged as a promising solution, where highly skewed  content popularity
is exploited to alleviate the heavy burden  on  backhaul networks and reduce latency in
content delivery \cite{CacheInAir,TechReprt_Cahing4G,5G_Cache_Architechture}.  Since popular contents vary at a time scale of several days \cite{ModelYoutube}, content placement can be performed every day during off-peak hours without causing an extra burden on the system. Compared with caching in wired
networks, the broadcast and superposition natures of the wireless medium make the optimal content placement in wireless networks a much more challenging problem and solving the problem has been the main theme in designing cache-enabled wireless systems and networks  \cite{LiveOnEdge}. Along the same theme, the current work considers caching for next-generation \emph{heterogeneous cellular networks} (HCNs), adopting the classic $K$-tier HCN model \cite{HCN-K-Tier}, and focuses on studying the optimal policy for placing contents in different BS tiers. 

\subsection{Related Work}
Extensive research has been conducted on studying the performance gain for joint
content-placement and wireless transmissions as well as designing relevant techniques.  From
the information-theoretic perspective, the capacity scaling laws were derived
for a large cache-enabled wireless network with a hierarchical tree structure
\cite{CachingWirelessNet}. In \cite{Caching_FundamentalLimits}, the novel idea of
integrating coding into user caching, called \emph{coded caching}, was proposed to improve
substantially the efficiency of content delivery over uncoded caching.
Specifically, exploiting joint coding of multiple files and the broadcast nature of downlink
channels, the content placement at BSs and delivery were jointly optimized to minimize the
communication overhead for content delivery. Coded caching in an erasure broadcast channel
was then studied in \cite{CachingErasureBroadcast} where the optimal capacity region has
been derived in some cases. In parallel, extensive research has also been carried out on the
more practical  uncoded caching where the focus is the design of strategies for
content-placement at BSs (or access points) to optimize  the network performance in terms of
the expected time for file downloading. Since optimal designs are NP-hard in
general~\cite{FemtoCaching_Distributed,FemtoCaching_BP}, most research has resorted to sub-optimal techniques with close-to-optimal performance. Specifically, practical algorithms have been designed for caching contents distributively at access points dedicated for content delivery using greedy algorithms \cite{FemtoCaching_Distributed}  and the theory of belief-propagation \cite{FemtoCaching_BP}. Furthermore, joint transmission and caching can  further improve the network performance \cite{Caching_Multirelay, CachingRelayNet_CrossLayer, CahicngMIMO_V.Lau}. Suboptimal solutions were developed  to maximize the quality of service for multi-relay networks \cite{Caching_Multirelay} and two-hop relaying network \cite{CachingRelayNet_CrossLayer} via decomposing the original problem into several simpler sub-problems. Considering the opportunistic cooperative MIMO, schemes were presented  in \cite{CahicngMIMO_V.Lau} to leverage multi-time-scale joint optimization of power and cache control to enable  real-time video streaming. Recent advancements in wireless caching techniques have been summarized  in various journal special issues and survey articles (see e.g., \cite{Caching_Misconception}).

It is also crucial to understand the performance gain that caching brings to large-scale
wireless networks. Presently, the common approach is to model and design cache-enabled
wireless networks using stochastic geometry. The approach leverages the availability of a
wide range of existing stochastic geometric network models, ranging from
\emph{device-to-device}~(D2D)~networks to HCNs,  and relevant results by adding caching
capacities to network nodes~\cite{CacheD2D_mobility,
CacheD2D_FiniteNetwork,CacheHWN_RelayUser,CacheHCNs_EURASIP,OptimalGeographicCaching,OptimalContentDis_Popularity,CooperativeCaching,CacheMulticast_HWN}.
In the resultant models, BSs and mobiles are typically distributed in the 2-dimensional (2D) plane as Poisson
point processes~(PPPs). Despite their similarity in the network nodes'~spatial
distributions, the cache-enabled networks differ from the traditional networks without
caching in their functions, with the former aiming at efficient content delivery and the latter
at reliable communication. Correspondingly, the performance of a cache-enabled network is
typically measured using a metric called \emph{hit probability}, defined as the probability
that a file requested by a typical user is not only available in the network  but can also 
be wirelessly delivered to the user~\cite{OptimalGeographicCaching}.  Based on
stochastic-geometry network models, the performance of cache-enabled D2D
networks~\cite{CacheD2D_mobility, CacheD2D_FiniteNetwork} and
HCNs~\cite{CacheHWN_RelayUser,CacheHCNs_EURASIP} were analyzed in terms of hit probability
as well as average throughput. For small-cell networks, one design challenge is that the
cache capacity limitation of BSs affects the  availability of  contents with low and
moderate popularity. A solution was proposed in~\cite{CooperativeCaching} for multi-cell
cooperative transmission/delivery in order to enhance the content availability.
Specifically, the proposed content-placement strategy is to partition the cache of each BS
into two halves for  storing both the most popular files and fractions of other files; then
multi-cell cooperation effectively integrates  storage spaces at cooperative BSs into  a
larger cache to increase content availability for improving the network hit probability.
Based on approximate performance analysis, the content-placement strategy derived in~\cite{CooperativeCaching}  is heuristic and the optimal one remains unknown. 

In the aforementioned work, the content placement at cache-enabled nodes is
\emph{deterministic}. An alternative strategy is  \emph{probabilistic (content) placement}
where a particular file is placed in the cache of  a network node (BS or mobile) with a
given probability \cite{OptimalGeographicCaching,OptimalContentDis_Popularity}, called
\emph{placement probability}. The strategy has also been considered in designing large-scale
cache-enabled networks \cite{OptimalGeographicCaching, OptimalContentDis_Popularity}. The
key characteristic of probabilistic  placement is that all files with nonzero placement
probabilities are available in a large-scale network with their spatial densities
proportional to the probabilities. Given its random nature,  the strategy fits the
stochastic-geometry models better than the deterministic counterpart as the former allows
for tractable analyses for certain networks as demonstrated in this work. The placement
probabilities for different content  files were optimized to maximize the hit probability for
cellular networks in \cite{OptimalGeographicCaching} and for D2D networks in
\cite{OptimalContentDis_Popularity}. It was  found therein  that the optimal placement
probabilities are highly dependent on, but not identical to, the (content) \emph{popularity
measures}, defined as the content-demand distribution over files as they are also functions
of network parameters, e.g., wireless-link reliability  and cache capacities. To improve
content availability, a hybrid scheme combining deterministic and probabilistic content
placement was proposed in \cite{CacheMulticast_HWN} for HCNs with multicasting where the
most popular files are cached at every macro-cell BS and different combinations of other
files are randomly cached at pico-cell BSs. Similar to the strategy in \cite{CooperativeCaching}, the
proposed strategy in \cite{CacheMulticast_HWN} does not lead to tractable
network-performance analysis and was optimized for the approximate hit probability. 

\subsection{Motivation, Contributions and Organization}
HCNs are expected to be deployed as next-generation wireless networks supporting content delivery besides communication and mobile computing \cite{LiveOnEdge}. In view of prior work, the existing strategies for content placement in large-scale HCNs are mostly heuristic and the optimal policies  in closed-form remain largely unknown, even though existing results reveal their various properties and dependence on network parameters.  This motivates the current work on analyzing the structure of the optimal  content-placement policies for HCNs. 

To this end, the cache-enabled HCN is modeled by adopting the classic $K$-tier HCN model for
the spatial distributions of BSs and mobiles \cite{HCN-K-Tier}. To be specific, the
locations of different tiers of BSs and mobiles are modeled as independent homogeneous PPPs
with non-uniform densities. Besides density, each tier  is characterized  by a set of
additional parameters including BS transmission power,  finite cache capacity and minimum
received \emph{signal-to-interference} (SIR) threshold required for successful content
delivery. Note that the use of SIR is based on the implicit assumption that the network is
interference limited. A user is associated with the nearest  BS where the requested file is
available. It is assumed that there exists a content database comprising $M$  files
characterized by corresponding popularity measures. Each user generates a random request for
a particular file based on the discrete  popularity distribution. In the paper, we propose a
tractable approach of probabilistic \emph{tier-level content placement (TLCP)} for the  HCN
where the placement probabilities are identical for all BSs belonging to the same tier but
are different across tiers.  The goal of the current work is to analyze the structure of the optimal policies for TLCP given the network-performance metric of hit probability.  The main contributions are summarized as follows.

\begin{enumerate}
  \item \textbf{Hit Probability Analysis}. By extending the results on outage probability for HCNs in~\cite{HCN-K-Tier}, the hit probability for cache-enabled HCNs are derived in closed form.  The results reveal that the metric is determined not only by the physical-layer related
  parameters,  including BS density, transmission power, and path-loss exponent, but also
  the content-related parameters,  including content-popularity measures  and placement
  probabilities. With uniform SIR thresholds for all tiers, the hit probability is observed
  to be a monotone increasing function of the placement probability and converges to a
  constant independent of BS density and transmission power as the placement probabilities
  approach $1$.

\item \textbf{Optimal Content Placement for Multi-Tier HCNs}. For a multi-tier HCN, the
  placement probabilities form an $M\times K$  matrix whose rows and columns correspond to the $M$ files
  and the $K$ tiers, respectively. First, consider a multi-tier HCN with uniform SIR thresholds
  for all tiers. Building on the results derived for  single-tier HCNs, a weighted sum (over tiers) of
  the placement probabilities for a particular file has the structure that it is \emph{proportional to the square root of the popularity measure with a fixed offset}. 
  Using this result, we derive the expressions for individual placement probabilities and
  reveal a useful structure allowing for a simple sequential computation of the
  probabilities. An algorithm is proposed to realize the aforementioned procedure. 
  Next, consider the general case of a multi-tier HCN with non-uniform SIR thresholds for
  different  tiers. In this case, finding the optimal content placement is non-convex and it
  is thus difficult to derive the optimal policy in closed-form. However, a sub-optimal
  algorithm can be designed leveraging the insights from the optimal policy structures for
  the previous cases. Our numerical results show that the performance of the proposed scheme
  is close-to-optimal.

\end{enumerate}

The remainder of the paper is organized as follows. The network model and metric are
described in Section II. The hit probability and optimal content placement for cache-enabled
HCNs are analyzed in Sections III and IV, respectively. Numerical results are provided in Section V followed by the conclusion in Section VI.

\section{Network Model and Metric}
In this section, we describe the mathematical model for the cache-enabled HCN illustrated
in~Fig.~\ref{Fig:Cache-enabled HCNs} and define its performance metric. The symbols used therein and their meanings are tabulated in Table~\ref{tab:table1}.

\begin{table}[t!]
  \centering
  \caption{Summary of Notations}
  \label{tab:table1}

  \begin{tabular}{cl}

    \toprule
    Symbol & Meaning \\
    \midrule

    $K$ & Total number of tiers in a HCN\\
    $M$ &  Total number of files in a datbase\\
    $\mathcal{F}_m$ &  The $m$-th file\\
    $\Phi_{k}$ & Point process of BSs in the $k$-th tier\\
    $\Phi_{mk}$, $\Phi_{mk}^c$ & Point process of BSs in the $k$-th tier \emph{with, without} file $\mathcal{F}_m$\\
    $\lambda_k$, $P_k$ & Density and transmission power of BSs in the $k$-th tier\\
    $\beta_k$  & SIR threshold of BSs in the $k$-th tier\\
    $h$  & Rayleigh fading gain with unit mean\\
    $\alpha$ & Path-loss exponent\\
    $C_k$ & Cache capacity of BSs in the $k$-th tier\\
    $q_m$ & Popularity measure for file $\mathcal{F}_m$\\
    $p_{mk}$ & Placement probability for file $\mathcal{F}_m$ in the k-th tier BSs \\
    
    \bottomrule

  \end{tabular}

\end{table}

\subsection{Network Topology}
The spatial distributions of BSs are  modeled using the classic $K$-tier
stochastic-geometry model for the HCN described as follows \cite{HCN-K-Tier}. The network
comprises $K$ tiers of BSs modeled as $K$ independent homogeneous PPPs distributed in the
plane. The $k$-th tier is denoted by $\Phi_{k}$ with the BS density  and transmission power
represented by $\lambda_{k}$ and $P_k$, respectively. Assuming an interference-limited network, the transmission by a BS to an
associated user is successful if the received SIR  exceeds a given threshold, denoted by
$\beta_k$, identical for all links in  the $k$-th tier.

We consider a particular frequency-flat channel, corresponding to a single frequency sub-channel of a broadband system.  Single antennas are deployed at all BSs and users. Furthermore, the BSs are assumed to transmit continuously  in the  unicast mode. The users are assumed to be Poisson distributed. As a result,  based on Slyvnyak's theorem \cite{StoGeoBook-Martin}, it is sufficient to consider in the network-performance analysis a typical user located at the origin, which yields the expected  experience for all users. The
channel is modeled in such a way that the signal power received at the user from a $k$-th tier BS located at $X_k \in \mathbb{R}^2$ is given by $P_{k}h_{X_k}
||X_{k}||^{-\alpha}$, where the random variable $h_{X_k}\sim \text{exp}(1)$ models
the Rayleigh fading and $\alpha>2$ is the path-loss exponent \footnote{In practice, the path-loss exponent may vary over the tiers. The corresponding conditional hit probability does not have a closed-form expression as in Lemma~3, resulting in an intractable  optimization problem. The current solution for the simpler case of uniform path-loss exponent can provide useful insights into designing practical content placement schemes for the said general case.}. Based on the channel model \footnote{The effect of shadowing on network performance is omitted in the current model for simplicity but  can be captured by modifying the model following the method in  \cite{ShadowingEffect}, namely appropriately scaling the transmission power of BSs in each tier. However, the corresponding modifications of the analysis and algorithmic design are straightforward without changing the key results and insights.}, the
interference power measured at the typical user, denoted by $I_0$, can be written as
 \begin{equation}
I_0 = \sum\nolimits_{k=1} ^{K} \sum\nolimits_{X \in \Phi_k\backslash X_k}P_k h_{X} ||X||^{-\alpha}
\label{Eq:Interf_org}, 
\end{equation}
where the  fading coefficients $\{h_X\}$  are assumed to be independent and identically distributed (i.i.d.).
\begin{figure}
\centering \includegraphics[width=8.5cm]{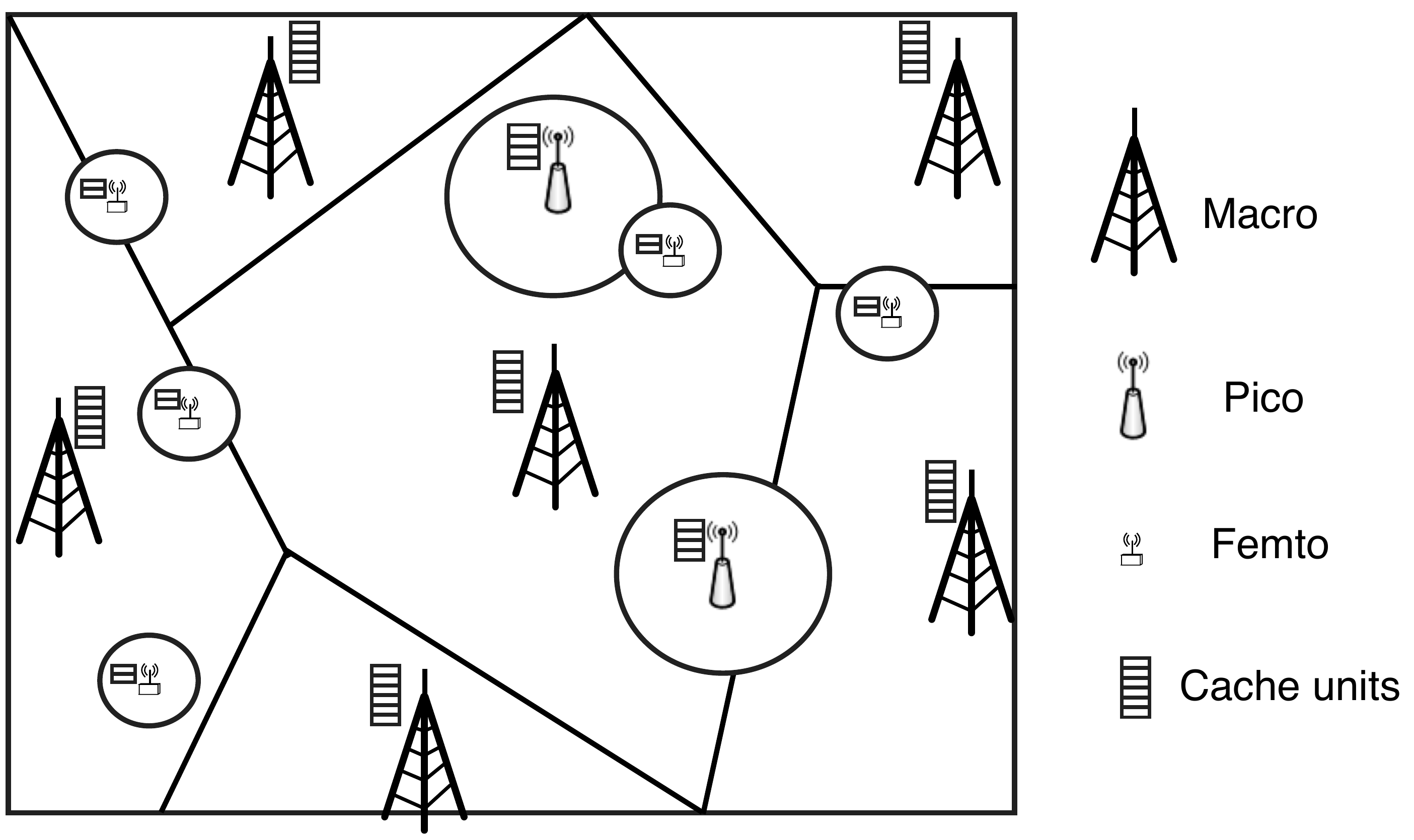}
\caption{A cache-enabled heterogenous cellular network.}\label{Fig:Cache-enabled HCNs}
\end{figure}

\subsection{Probabilistic Content Placement}\label{Section:PlaceModel}
In this paper, we consider a content (such as video file) database containing $M$ files with normalized size equal to $1$ following the literature  \cite{OptimalGeographicCaching,CooperativeCaching,CacheMulticast_HWN}.\footnote{In practice, a large content file to be cached is usually divided in to units of equal sizes because they have  different popularity. For instance, the beginning $1$-minute of a YouTube video is much more popular than the remainder. Thus, to be precise, the equal-size files considered in this paper correspond to content units in practice.} As
illustrated in Fig.~\ref{Fig:Cache-enabled HCNs}, the BSs from different tiers are assumed
to have different cache capacities which are denoted by $C_k$ for the $k$-th tier with $k =
1, 2, \cdots, K$.  We make the practical assumption that not all BSs have sufficient
capacities for storing the whole database, i.e., $C_k \leq M, \forall k$. We adopt a probabilistic content placement
scheme similar to the one in \cite{OptimalGeographicCaching} to randomly select files for
caching at different tiers under their cache capacity constraints: 
\begin{equation}\label{Eq:Cap:Const}
\sum_{m=1}^{M}p_{mk}\leq C_{k},   \forall k.  
\end{equation}
Specifically, the $m$-th
file, denoted by  $\mathcal{F}_m$, is cached at a tier-$k$ BS with a fixed probability
$p_{mk}$ called a \emph{placement probability}. The placement probabilities, $(p_{1k}, p_{2k}, \cdots, p_{Mk})$,
are identical for all BSs in the same tier $k$, $k=1,\ldots,K$. They specify  the  \emph{tier-level content placement} (TLCP). Grouping the placement probabilities yields the following \emph{placement probability matrix}:
\begin{align}
\qquad {\bf P} = \;\left[ {\begin{array}{*{20}{c}}
 p_{11}&p_{12}& \cdots &p_{1K}\\
p_{21}&p_{22}& \cdots &p_{2K}\\
 \vdots & \vdots  &{}& \vdots \\
p_{M1}&p_{M2}& \cdots &p_{MK}
\end{array}} \right].  \label{Eq:Place:Matrix}
\end{align}
The rows and columns of $\bP$ correspond to different files and different tiers, respectively. Given the placement probabilities in $\bf P$ and under the cache-capacity constraints in \eqref{Eq:Cap:Const},  there exist specific  strategies  of randomly placing contents at individual  BSs such that File $\mathcal{F}_m$ is available at a tier-$k$ BS with a  probability exactly equal to  $p_{mk}$\cite{OptimalGeographicCaching}. One  of such strategies   is illustrated in \cite[Fig. 1]{OptimalGeographicCaching}. Given the existence of random-placement strategies for achieving the content availability specified by $\bP$, this paper focuses on optimizing $\bP$ for maximizing the hit probability.

The files in the content database differ in popularity, measured by a corresponding set of values $\{q_m\}$ with $q_m \in [0,1]$ for all $m$ and $\sum_{m=1}^{M}q_{m}=1$ \cite{OptimalGeographicCaching,OptimalContentDis_Popularity,CooperativeCaching,CacheMulticast_HWN}. This set is a probability mass function such that the typical user requests file $\mathcal{F}_m$ with probability $q_m$. Without loss of generality, it is assumed that the files are ordered in
decreasing popularity, i.e., $q_1>q_2>\cdots>q_M$.

\subsection{Content-Centric Cell Association}
Content-centric cell association accounts for both the factor of link
reliability and the factor of content availability. We adopt a common scheme that associates a user with
the BS that maximizes the received signal power among those having the requested file   (see e.g., \cite{CacheMulticast_HWN, Cui_SingleTierCacheNetwork}).\footnote{\emph{Coordinated multiple access point} (CoMP) defined in the LTE standard can be applied  to improve the network performance via associating each  user with multiple BSs. Adopting the technology in the current network model does not lead to tractable analysis. However, it is possible to develop practical content-delivery schemes for HetNets with CoMP by integrating  the current optimal TL content placement and the design of cooperative content delivery in \cite{CooperativeCaching}.} It is important to note that due to limited BS storage, the database cached at BSs  is only the popular subset of all contents. Thus, it is possible that a file requested by a user is unavailable at the network edge, which has to be retrieved from a data center across the backhaul network. In such cases, the classic cell association rule is applied to connect the user to the nearest BS. These  cases occur infrequently and furthermore are outside the current scope of content placement at the network edge. Thus,  they are omitted in our analysis for simplicity following the common approach in the literature (see e.g., \cite{Cui_SingleTierCacheNetwork}). For ease of exposition, we partition 
the HCN into $M\times K$ effective tiers, called the \emph{content-centric tiers},
according to the file availability within each tier.  
The $(m, k)$-th content-centric tier refers to the process of tier-$k$ BSs with file
$\mathcal{F}_m$, denoted by $\Phi_{mk}$, while the remaining tier-$k$ BSs are denoted by
$\Phi_{mk}^c$ with $\Phi_{mk}\cup \Phi_{mk}^c = \Phi_{k}$. Due to the probabilistic
content placement scheme, $\Phi_{mk}$ and  $\Phi_{mk}^c$ are independent PPPs with densities
$p_{mk}\lambda_k$ and $(1 - p_{mk})\lambda_k$, respectively. A user is said to be
associated with the $(m,k)$-th content-centric tier if the user requests $\mathcal{F}_m$ and
is served by a tier-$k$ BS. Then, conditioned on the typical user requesting file
$\mathcal{F}_m$, the serving BS $X_k$ is given by
\begin{equation}
\textrm{(Cell Association)}\ \ X_k =\arg \max_{X\in\bigcup_k \Phi_{mk}} P_X  ||X||^{-\alpha},
\end{equation}
where $P_X$ denotes BS~$X$'s transmission power. In addition, conditioned on the typical
user requesting file $\mathcal{F}_m$, the interference power ${I_0}$ in \eqref{Eq:Interf_org} can be written in terms of the content-centric tiers as:
\begin{equation}\label{Eq:I:Cond}
\begin{aligned}
{I_0}(\mathcal{F}_m) &=\sum\nolimits_{k=1}^{K} \sum\nolimits_{X \in \Phi_{mk} \setminus X_k} P_k h_X ||X||^{-\alpha}  +\\ &\qquad \sum\nolimits_{k=1}^{K} \sum\nolimits_{X \in \Phi_{mk}^{c}} P_k h_X ||X||^{-\alpha} .
\end{aligned}
\end{equation}

\subsection{Network Performance Metric}
The network performance is measured by the \emph{hit probability}  defined as the
probability that a file the typical user requested is not only cached at a BS but also
successfully delivered by the BS over the wireless channel (see e.g.,
\cite{OptimalGeographicCaching}). By definition, the hit probability quantifies the reduced fraction of backhaul load.  In addition, it can also indicate the reduction of mean latency in the backhaul network (see Appendix \ref{App:P_Delay} for details). Therefore, we use the hit probability as the main network performance metric in this paper. For the
purpose of analysis,  let $\mathcal{P}$ denote the (unconditional) hit probability, $\mathcal{P}_{m}$ denote  the conditional hit probability
given that the typical user requests file $\mathcal{F}_m$, and $q_m$ denote the content popularity for file $\mathcal{F}_m$. Then
\begin{equation}
\mathcal{P}=\sum_{m=1}^{M} q_{m} \mathcal{P}_{m}.  \label{Eq:HitProb_Def}
\end{equation}
Furthermore, define the \emph{association probability} indexed by $(m, k)$, denoted by
$A_{mk}$, as the probability that the
 typical user is associated with the $(m, k)$-th content-centric tier. The hit probability
 conditional on this event is represented by $\mathcal{P}_{mk}$. It follows that
 \begin{equation}
\mathcal{P}_{m} = \sum_{k=1}^{K} A_{mk} \mathcal{P}_{mk}.  \label{Eq:HP_File_def}
\end{equation}

\section{Analysis of Hit Probability }
In this section, the hit probability for  the cache-enabled HCN is calculated.   To this end, the association probabilities and the probability density function (PDF) of the serving distances are derived in the following two lemmas, via directly modifying Lemmas~$3$ and $5$ in  \cite{FlexibleCellAssociation} enabled by the interpretation of the HCN as one comprising $M\times K$ content-centric  tiers (see Section~\ref{Section:PlaceModel}).

\begin{lemma} [Association Probabilities] \label{Lemma:Association probability}
\emph{The association  probability that the typical user belongs to the $(m, k)$-th effective tier is given as
\begin{equation}
A_{mk}=\frac{p_{mk}\lambda_{k}P_{k}^{\delta}}{\sum_{j=1}^{K}p_{mj}\lambda_{j}P_{j}^{\delta}},  \label{eq:AssocProb}
\end{equation}
where the constant $\delta = \frac{2}{\alpha}$.
}
\end{lemma}
\begin{proof}
See Appendix \ref{Pf: Association-Probability}. 
\end{proof}
The result in Lemma~\ref{Lemma:Association probability} shows that the typical user requesting
a particular file is more likely to be associated with one of those tiers having not only
larger placement probability but also  denser  BS or higher BS transmission power, aligned with
intuition. In addition, it is shown that if $\delta$ is small, the placement probability and BS density have more
dominant effects on determining the association probability than transmission power, since $P_{j}^{\delta}$ converges to one for all $j$ as $\delta \rightarrow 0$. 

\begin{lemma}[Statistical Serving Distances]\label{Lemma: Serving distance}
\emph{
The PDF of the serving distance between the typical user and the associated BS in the $(m, k)$-th effective tier is given as
\begin{equation}
f_{R}(r)\!=\!\frac{2\pi p_{mk}\lambda_{k}}{A_{mk}}r\exp\left(
-\pi\sum_{j=1}^{K}p_{mj}\lambda_{j}\left(\frac{P_{j}}{P_{k}}\right)^{\delta}r^{2}\right), \label{eq:pdfServingDis}
\end{equation}
where $A_{mk}$ is given in  \eqref{eq:AssocProb}.
}
\end{lemma} 

Next, we are ready to derive the hit probabilities using Lemmas~\ref{Lemma:Association
probability} and \ref{Lemma: Serving distance}. For ease of notation, we define
the following two functions $Q(\beta_k)$ and $V(\beta_k)$, which are related to the interference coming from the BSs \emph{with} and \emph{without} the file $\mathcal{F}_m$, respectively:
\begin{align}
Q(\beta_k) &=\frac{\delta \beta_k}{1-\delta} \;{}_{2}F_{1}[1,1-\delta;2-\delta;-\beta_k], \label{Eq:Q_Fun}\\
V(\beta_k) &=\beta_k^\delta \; \delta \pi \csc(\delta \pi),\label{Eq:V:Fun}
\end{align}
where  \!${}_{2}F_{1}[\cdot]$\! denotes the Gauss hypergeometric function and $\csc(.)$ is the cosecant-trigonometry function. To further simplify the expression of hit probability, we define the following function:
\begin{equation}
W(\beta_k) = 1+Q(\beta_k)-V(\beta_k). \label{Eq:T:Fun}
\end{equation}
Then the conditional hit probability can be written as shown in the following lemma.
\begin{lemma}[Conditional Hit Probability]\label{Lem:ConditionalHitProb}
\emph{In the  cache-enabled HCN, the conditional hit probability for the typical user requesting file $\mathcal{F}_{m}$ is given as
\begin{equation}
\mathcal{P}_{m}=\sum_{k=1}^{K} \frac{p_{mk} \lambda_k
P_k^{\delta}}{W(\beta_k)\sum\limits_{i=1}^{K}p_{mi} \lambda_i P_i^{\delta}+V(\beta_k)
\sum\limits_{i=1}^{K}\lambda_i P_i^{\delta}}, \label{Eq:ConHitProb}
\end{equation}
where the functions $V(\cdot)$ and $W(\cdot)$ are defined in \eqref{Eq:V:Fun} and \eqref{Eq:T:Fun}, respectively. 
}
\end{lemma}
\begin{proof}
See Appendix \ref{Pf:ConditionalHitProb}.
\end{proof}
Using Lemma~\ref{Lem:ConditionalHitProb} and the definition of hit probability in
\eqref{Eq:HitProb_Def}, we obtain the first main result of this paper.
\begin{theorem}[Hit Probability] \label{Thm:HitProb}
\emph{
The hit probability for the  cache-enabled HCNs is given as
\begin{equation}\nn
\mathcal{P} =\sum_{m=1}^{M} q_{m} \sum\limits_{k=1}^{K} \frac{p_{mk} \lambda_k
P_k^{\delta}}{W(\beta_k)\sum\limits_{i=1}^{K}p_{mi} \lambda_i P_i^{\delta}+V(\beta_k)
\sum\limits_{i=1}^{K}\lambda_i P_i^{\delta}},
\end{equation}
where functions $V(\cdot)$ and $W(\cdot)$ are given in \eqref{Eq:V:Fun} and \eqref{Eq:T:Fun}, respectively. 
}
\end{theorem}
Theorem \ref{Thm:HitProb} shows that the hit probability is determined by two sets of
network parameters: one set is related to the physical layer including the BS density
$\{\lambda_k\}$, transmit power $\{P_k\}$, and path-loss parameter $\delta$; the other set contains content-related parameters including the popularity measures $\{q_m\}$ and placement probabilities $\{p_{mk}\}$.

From Theorem \ref{Thm:HitProb}, we can directly obtain hit probabilities for two
special cases, namely, the single-tier HCNs and the multi-tier HCNs with uniform SIR
thresholds, as shown in the following two corollaries. 

\begin{corollary}[Hit Probability for Single-Tier HCNs] \label{CR:SpecialCase:SingleTier}
\emph{
Given $K=1$, the hit probability for cache-enabled HCNs is 
\begin{equation}\label{SCDP_SigleTier}
\mathcal{P}=\sum_{m=1}^{M} q_{m} \frac{p_m}{W(\beta) p_m + V(\beta)},
\end{equation}
where the functions $V(\cdot)$ and $W(\cdot)$ are given in \eqref{Eq:V:Fun} and \eqref{Eq:T:Fun}. 
}
\end{corollary}

Corollary \ref{CR:SpecialCase:SingleTier} shows that the hit probability for single-tier
cache-enabled networks is independent with BS density and transmit power, which is a well-known characteristic of interference-limited cellular networks. On the other hand, it is found to be monotone increasing with growing placement probabilities as the spatial content density increases.

\begin{corollary}[Hit Probability for Multiple-tier HCNs with Uniform SIR Thresholds] \label{CR:SpecialCase:IdenticalSIR}
\emph{
Given $\beta_{k}=\beta\ \forall k$, the hit probability for the cache-enabled HCNs is given as
\begin{align}
&\mathcal{P}=\sum_{m=1}^{M} q_{m}
\frac{\sum_{k=1}^{K}p_{mk}\lambda_{k}P_{k}^{\delta}}{W(\beta)
\sum_{k=1}^{K}p_{mk}\lambda_{k}P_{k}^{\delta} + V(\beta)
\sum_{k=1}^{K}\lambda_{k}P_{k}^{\delta} },
\end{align}
where functions $V(\cdot)$ and $W(\cdot)$ are given in \eqref{Eq:V:Fun} and \eqref{Eq:T:Fun}, respectively. 
}
\end{corollary}

\begin{remark}[Effects of Large Cache Capacities] \label{RM: IdenticalSIR}
\emph{Corollary \ref{CR:SpecialCase:IdenticalSIR} shows that the hit probability is a
monotone increasing function of the placement probabilities $\{p_{mk}\}$ and converges to a
constant, which is independent of the BS densities and transmission powers, as all the placement probabilities become
ones, corresponding to the case of large cache capacities. At this limit, the cache-enable
HCN is effectively the same as a traditional interference-limited HCN for general data
services and the said independence is due to a uniform SIR threshold and is well known in the literature (see e.g., \cite{HCN-K-Tier}).
}
\end{remark}

\section{Optimal Tier-Level Content Placement}

In this section, we maximize the hit probability derived for the cache-enabled HCNs in the
preceding section over the placement probabilities. 

\subsection{Problem Formulation}
The TLCP problem consists of finding the placement matrix $\bP$ in
\eqref{Eq:Place:Matrix} that maximizes the hit probability for HCNs as given in Theorem~\ref{Thm:HitProb}. Mathematically, the optimization problem can be formulated as follows: 
\begin{equation}\tag{$\textbf{P0}$}\label{Opt_Org}
\begin{aligned}
\mathop{\max}\limits_{\bf{P}}   & \sum_{m=1}^{M} q_{m} \sum_{k=1}^{K} \frac{p_{mk} \lambda_k P_k^{\delta}}{W(\beta_k)\sum_{i=1}^{K}p_{mi} \lambda_i P_i^{\delta}+V(\beta_k) \sum_{i=1}^{K}\lambda_i P_i^{\delta}} \\
\textmd{s.t.} &\sum_{m=1}^{M}p_{mk}\leq C_{k},   \forall k,  \\
& p_{mk}\in [0,1],  \forall m, k,
\end{aligned}
\end{equation}
where the first constraint from \eqref{Eq:Cap:Const} is based on the BS cache capacity  for each tier and the second constraint arises from the fact that  $p_{mk}$ is a probability. 

It is numerically difficult to directly solve Problem P0, since it has a structure of ``sum-of-ratios" with a non-convex nature and has been proved to be NP-complete. In order to provide useful insights and results for tackling the problem, the optimal content placement policies  are first analyzed for the special case of single-tier HCNs and then extended to multi-tier HCNs.

\subsection{Single-Tier HCNs}
For the current case with $K=1$, using Corollary~\ref{CR:SpecialCase:SingleTier},  Problem P0 is simplified as:
\begin{equation}\tag{$\textbf{P1}$}\label{Opt_SingleTier}
\begin{aligned}
\mathop {\max}\limits_{\bf{p}} \; &{\sum_{m=1}^{M} q_{m} \frac{p_m}{W(\beta) p_m + V(\beta)} } \\
{\textmd{s.t.}}\;\;&{\sum_{m=1}^{M}p_{m}\leq C},  \\
& p_{m}\in [0,1],  \forall m,
\end{aligned}
\end{equation}
where $p_m$ denotes the placement probability for file $\mathcal{F}_m$, the vector $\mathbf{p}=(p_1,p_2,\cdots,p_m)$, $C$ denotes the cache capacity for single-tier HCNs. 

It should be noted that Problem P1 has the same structure as that for the asymptotic (high SNR and high user density) case in \cite{Cui_SingleTierCacheNetwork} where the single-tier cache-enabled BSs are distributed as a PPP and random combination of files are cached in each BS with probability $p_m$. Nevertheless, it is still valuable to discuss this special case since it provides useful insights for solving the complex problem for multi-tier HCNs. Therefore, this paper focuses on these insights.

Problem P1 is convex since the objective function is convex and the constraints are  linear and can thus be solved using the Lagrange method. The  Lagrangian function can be written as
\begin{equation}\label{Lagrangian_SingleTier}
L(\mathbf{p},u)\!=\!\sum_{m=1}^{M} q_{m} \frac{p_m}{W(\beta) p_m\!+\!V(\beta)} + u\!\left(C\!-\!\sum_{m=1}^{M}p_{m}\!\right),
\end{equation}
{where $u \geq 0$ denotes the Lagrangian multiplier. Using the Karush-Kuhn-Tucker (KKT)
condition, setting the derivative of  $L$ in  \eqref{Lagrangian_SingleTier} to zero leads to
the optimal placement probabilities as shown in   Theorem \ref{OptimalCaching_SingleTier}
where the optimal Lagrange multiplier is denoted by $u^*$. Note that the capacity constraint is active at the optimal point, namely  $\sum_{m=1}^M p_m^*(u^*)=C$. This result comes from the fact that the objective function of Problem P1 is a monotone-increasing function of $\{p_m\}$.

\begin{theorem}[Optimal TLCP for Single-Tier HCNs]\label{OptimalCaching_SingleTier}
\emph{For the single-tier cache-enabled HCN, given the optimal Lagrangian multiplier  $u^*$,
the optimal content placement  probabilities, denoted by $\{p^*_m\}$, that solve Problem~P1  are given as
\begin{equation}
p_{m}^{*}(u^*)=\begin{cases}
1, & q_m \geq T_1,\\
\frac{\sqrt{V(\beta)}}{\sqrt{u^*}W(\beta)}\sqrt{q_m} - \frac{V(\beta)}{W(\beta)}, & T_0 < q_m < T_1, \\
0, & q_m \leq T_0,
\end{cases}
\end{equation}
where the thresholds  $T_1=\frac{u^* (W(\beta)+V(\beta))^2}{V(\beta)}$ and $T_0=u^* V(\beta)$,  and 
the optimal Lagrange multiplier $u^*$ satisfies the equality
\begin{equation}
\sum_{m=1}^M p_m^*(u^*)=C.
\end{equation}
}
\end{theorem}

In addition, the optimal Lagrangian multiplier $u^*$ in Theorem~\ref{OptimalCaching_SingleTier}  can be found via a simple bisection search. Let $D$ be the number of iterations needed to find the optimal Lagrange multiplier $u^*$. Clearly, the computational complexity of TLCP for single-tier HCNs is $O(DM)$. The corresponding algorithm is shown in Algorithm~\ref{tab:SingleTier}.

\begin{algorithm}[t!]
  \centering
  \caption{Computing the Optimal Lagrangian Multiplier $u^*$ by a Bi-section Search.}
  \label{tab:SingleTier}
  \begin{tabular}{l}
    initialize $u^0 \in [u^{(0,\min)}, u^{(0,\max)}]=[\frac{q_M V}{(W+V)^2}, \frac{q_1}{V}]$\\
    repeat\\
    \;\;\;\; $u^{(\ell+1)}=u^{(\ell, \min)}+\frac{u^{(\ell, \max)}-u^{(\ell, \min)}}{2}$\\
    \;\;\;\; if $\sum_{m=1}^{M} p_m(u^{(\ell+1)})<C$,  \;\;\;\; $u^{(\ell+1,\max)}=u^{(\ell+1)}$\\
    \;\;\;\; else \;\;\;\; $u^{(\ell+1,\min)}=u^{(\ell+1)}$\\
    until $u$ converges\\
  \end{tabular}
\end{algorithm}

\begin{remark}[Offset-Popularity Proportional Caching Structure]\label{FundamentalPolicyStructure}
\emph{As illustrated in Fig. \ref{Fig:OPP Structure}, the optimal content placement in
Theorem \ref{OptimalCaching_SingleTier} has the mentioned  \emph{offset-popularity proportional}
(OPP) structure described as follows. Specifically, if the popularity measure of a particular file is
within the range $[T_0, T_1]$, the optimal placement probability, $p_m$, monotonically increases with the \emph{square root} of the popularity measure, i.e., $\sqrt{q_m}$. Otherwise, the probability is either $1$ or $0$ depending on whether the measure is above or below the range. Furthermore, the probability is offset by a function $V(\beta)/W(\beta)$ of the SIR threshold and scaled by a function of both the threshold and the cache capacity $C$. }
\end{remark}

\begin{figure}
\centering \includegraphics[width=8cm]{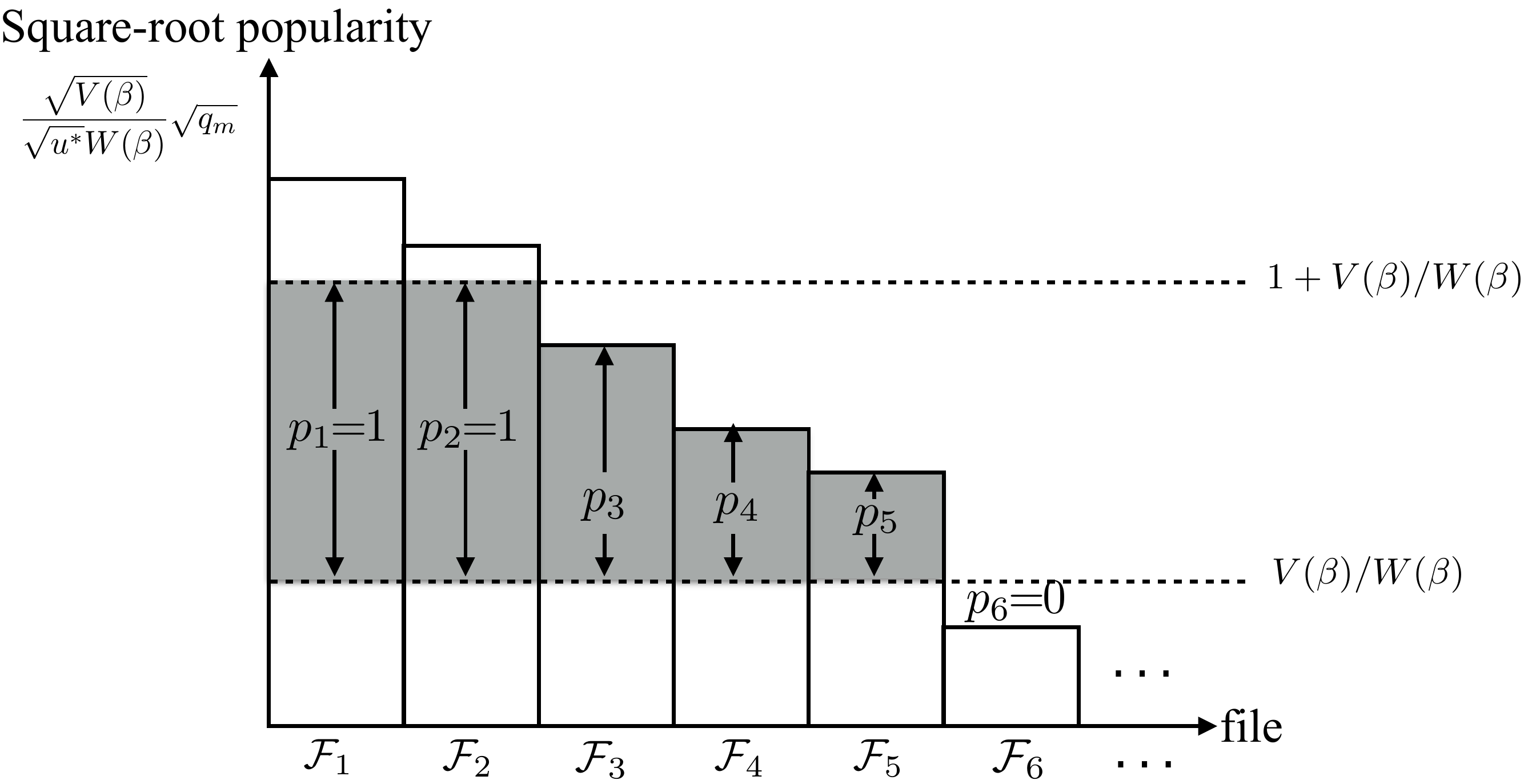}
\caption{The structure of the optimal content-placement policy for single-tier HCNs.}\label{Fig:OPP Structure}
\end{figure}

\begin{figure}
\centering \includegraphics[width=8cm]{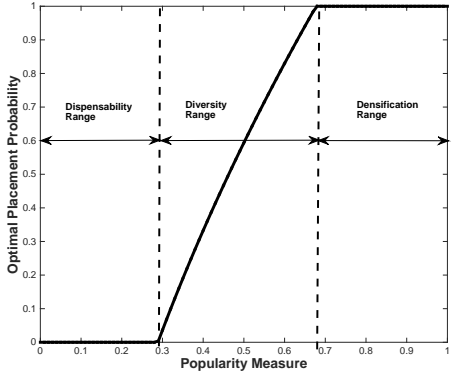}
\caption{The  effects of content popularity on the optimal placement probability.}\label{Fig:ContentMUL_DIV}
\end{figure}

\begin{remark}[Effects of Content Popularity on Optimal Placement Probability] \label{MultiplexityDiversity}
\emph{The result in Theorem~\ref{OptimalCaching_SingleTier} shows that the optimal content placement probability is decided by its popularity measure. In particular, content files can be
separated by defining three ranges of popularity measure, corresponding to placement
probabilities of   $0$, $(0, 1)$ and $1$ as illustrated in Fig.~\ref{Fig:ContentMUL_DIV},
called the \emph{dispensability}, \emph{diversity}, and \emph{densification} ranges, respectively. In the dispensability range, the files are highly unpopular and do not need to be cached in the network. In contrast, the files in the densification range are highly popular such that their spatial density should be maximized by caching the files at every BS. Last, files in the diversity range have moderate popularity and it is desirable to have all of them available in the network, corresponding to enhancing spatial content diversity. As a result,  they are cached at different fractions of BSs. 
}
\end{remark}

\begin{remark}[Effects of SIR threshold]\label{OptCaching_SingleTier_SIR}
\emph{ The SIR threshold $\beta$ affects both  the popularity thresholds ($T_0$ and $T_1$) in the  optimal placement policy (see Theorem~\ref{OptimalCaching_SingleTier}). It is observed from numerical results that both thresholds are monotone increasing functions of $\beta$. }
\end{remark}

\begin{remark}[Effects of Lagrangian multiplier $u^*$]\label{Multiplier_SingleTier}
\emph{ The value of Lagrangian multiplier affects the popularity thresholds $T_0$ and $T_1$, and is determined by the capacity constraint equality, i.e., $\sum_{m=1}^M p_m^*(u^*)=C$. In the case that the requested cache unit is larger than the cache capacity, i.e., $\sum_{m=1}^M p_m(u)>C$, the Lagrangian multiplier $u$ should be increased to enlarge the popularity thresholds and thus decrease the placement probabilities, and vice versa.}
\end{remark}

Problem P1 is considered purposely to help solve the general version and also for clarity in exposition. In particular, the insight from solving P1 is exploited to solve P2 in closed form given uniform SIR thresholds and develop a sub-optimal scheme for the case with non-uniform thresholds.

\subsection{Multi-tier HCNs with Uniform SIR Thresholds}
Consider a  multi-tier HCN with uniform SIR thresholds for all tiers.  Based on the hit probability in Corollary \ref{CR:SpecialCase:IdenticalSIR}, Problem P0 for the current case is given as:
\begin{equation}\tag{$\textbf{P2}$}
\begin{aligned}
\mathop {\max }\limits_{\bf{p}} \; &{\sum_{m=1}^{M} q_{m} \frac{\sum_{k=1}^{K}p_{mk}z_k}{W \sum_{k=1}^{K}p_{mk}z_k + V^{\prime} }} \\
{\textmd{s.t.}}\;\;&{\sum_{m=1}^{M}p_{mk}\leq C_k}, \forall k,  \\
& p_{mk}\in [0,1],  \forall m,k, 
\end{aligned}
\end{equation}
where $z_k$ and $V'$ are  constants defined as  $z_k=\lambda_k P_k^{\delta}$ and $V^{\prime} = V \sum_{k=1}^{K} z_k$.
One can see that the problem is convex and can thus be solved numerically using a standard convex-optimization solver. However, the numerical approach may have high complexity  if the content database is large and further yields little insight into the optimal policy structure. Thus, in the remainder of this section,  a simple algorithm is developed for sequential computation of the optimal policy, which also reveals some  properties of the policy structure. 

To this end,  define the tier-wise weighted sum of placement probabilities for each file as 
\begin{equation}
g_{m}=\sum_{k=1}^{K} p_{mk} z_k, \qquad m = 1, 2, \cdots, M. 
\end{equation}
Using this definition, a relaxed version of  Problem P2 can be rewritten as follows: 
\begin{equation}\tag{$\textbf{P3}$}\label{Opt_MultiTier_SpecialCase}
\begin{aligned}
\mathop {\max }\limits_{\{g_m\}} \; &{\sum_{m=1}^{M}  \frac{q_{m} g_m}{W g_m + V^{\prime} }} \\
{\textmd{s.t.}}\;\;&\sum_m g_m \leq \sum_k C_k z_k, \forall k,  \\
& 0\leq g_m \leq \sum_k z_k,  \forall m. 
\end{aligned}
\end{equation}
Comparing Problem P3 with P1 for the single-tier HCNs, one can see the two problems have identical forms. Thus, this allows Problem P3 to be solved following a similar procedure as P1, yielding the following proposition. 
\begin{proposition}[Weighted Sum of Optimal Placement Probabilities]\label{Prop:SumWeighted}
\emph{The weighted sum of the optimal placement probabilities for multi-tier HCNs with
uniform SIR thresholds, denoted by $g_m^*$, is given as:
\begin{equation}\label{Eq:WeightedSum}
g_{m}^{*}(\eta^*)\!=\!\begin{cases}
\sum_{k=1}^{K} \lambda_k P_k^{\delta}, & \text{if} \;q_m \geq T_1^{\prime},\\
(\sqrt{q_m V^{\prime}/\eta^*}-V^{\prime})/W, & \text{if}\; T_0^{\prime} < q_m <
T_1^{\prime},\!\! \\
0, & \text{if} \; q_m \leq T_0^{\prime},
\end{cases}
\end{equation}
where $T'_1=\frac{\eta^* (W^{\prime}+V^{\prime})^2}{V^{\prime}}$, $T'_0=\eta^* V^{\prime}$, $W^{\prime}=W \sum_{k=1}^K z_i$ and the optimal Lagrange multiplier $\eta^*$ satisfies the following equality
\begin{equation}\label{OptMultiplier}\nn
\sum_{m=1}^{M} g_{m}^{*}(\eta^*)=\sum_{k=1}^{K}C_k z_k.
\end{equation}
}
\end{proposition}

The value of $\eta^*$ can be found using the bisection search in Algorithm~1. Then the optimal values for the weighted sum $\{g_{m}^{*}\}$ can be computed using Proposition~\ref{Prop:SumWeighted}.

Problem P3 is the relaxed version of P2 since  the feasible region of P3 is larger than that of P2. Let $\{p^*_{mk}\}$ denote the optimal placement probabilities solving Problem P2 and $\{g_m^*\}$ the weighted sums solving Problem P3. The following proposition shows that the relaxation does not compromise the optimality of the solution.

\begin{proposition}\label{Prop_WeightedSum}
\emph{The solution of Problem P3 solves P2 in the sense that  $\sum_{k=1}^{K} p_{mk}^*
z_k=g_{m}^*,  m = 1, 2, \ldots, M$. 
}
\end{proposition}
\begin{proof}
See Appendix \ref{Pf:OptimalCachingProb_IdenticalBeta}.
\end{proof}

Next, based on the results in Propositions~\ref{Prop:SumWeighted} and \ref{Prop_WeightedSum}, the structure of the optimal placement policy is derived as shown in Theorem~\ref{Tm:OptimalSetting}, which enables low-complexity sequential computation of the optimal placement probabilities. 

\begin{theorem}[Sequential Computation of Optimal Placement Probabilities]\label{Tm:OptimalSetting}
\emph{One possible policy for optimal TLCP for the HCNs with uniform SIR thresholds is given as follows: 
\begin{equation}\label{Eq:OptimalCacheProb_MultiTier}
p_{mk}^{*}\!\!=\!\!\begin{cases}
1, \!\!\!\!& \text{if} \;q_m \geq T_1^{\prime}, \\
\!\!\min\!\l(\frac{1}{z_k}\!\displaystyle{\sum_{j=1}^k} \zeta_{mj}^* z_j\!-\! \frac{1}{z_k}\sum_{j=1}^{k-1} p_{mj}^* z_j,1\!\r), \!\!\!\!& \text{if}\;  T_0^{\prime} \!<\! q_m \!<\! T_1^{\prime},\!\! \\
0,\!\!\!\! & \text{if} \; q_m \leq T_0^{\prime}, 
\end{cases}
\end{equation}
where 
\begin{equation}
\zeta_{mj}^*=\frac{g_m^*}{ \sum_{i=m}^{M}g_i^*}\l(C_k-\sum_{i=1}^{m-1} p_{ik}^*\r), 
\end{equation}
and $g_m^*$ is as given in Proposition~\ref{Prop:SumWeighted}. 
}
\end{theorem}
\begin{proof}
See Appendix \ref{Pf:OptimalSetting}.
\end{proof}
A key observation of the policy structure in Theorem~\ref{Tm:OptimalSetting} is that $p_{mk}^{*}$ depends only on  $\{p_{ij}^{*}\}$ with $i <  m$ and $j < k$. This suggests that the optimal placement probabilities can be computed \emph{sequentially} as shown in Algorithm~$2$ and thus the computational complexity of Algorithm 2 is $O(DMK)$, where $D$ is the number of iterations needed to find the optimal Lagrange multiplier.

One can observe from Proposition~\ref{Prop:SumWeighted} that the optimal solution for Problem P2 is not unique. In other words, there may exist a set of placement probabilities different from that computed using Algorithm~$2$ but achieving the same hit probability. 

\begin{algorithm}[t!]
  \centering
  \caption{Sequential Computation of Optimal Placement Probabilities for Multi-tier HCNs
  with  Uniform SIR Thresholds.}
  \label{tab:PropAlloc}
  \begin{tabular}{l}
    1. Compute $\eta^*$ using Algorithm~$1$ and $\{g_m^*\}$ using Proposition~\ref{Prop:SumWeighted}\\
    2. For $m=1:M$\\
        \;\;\;\;\;\; for $k=1:K$\\
        \;\;\;\;\;\;\;\;\;\; set $p_{mk}^{*}$ according to \eqref{Eq:OptimalCacheProb_MultiTier}\\
        \;\;\;\;\;\;\;\;\;\; update $C_k^{\prime}=C_k^{\prime}-p_{mk}^{*}$\\
        \;\;\;\;\;\;\;\;end\\
        \;\;\;\;\;end\\
  \end{tabular}
\end{algorithm}

\subsection{Multi-tier HCNs with Non-Uniform SIR Thresholds}

For the current  case, the problem of  optimal content placement is Problem P0. As the problem is non-convex, it is  numerically complex to solve and also difficult to develop low-complexity algorithms by analyzing  the optimal policy structure. Therefore, a low-complexity sub-optimal algorithm is proposed for content placement for the current case. The algorithm is designed based on approximating the hit probability in Theorem~\ref{Thm:HitProb} by  neglecting the effects of the  placement probability of other tiers on the hit probability of the $k$-th tier. Specifically, given $z_i = \lambda_i P_i^\delta$ as defined previously and by replacing the term $\sum_i p_{mi}z_i$ with $p_{mk}z_k$, the hit probability in Theorem~\ref{Thm:HitProb} can be approximated by  $\widetilde{\mathcal{P} }$ given as 
\begin{align}
\widetilde{\mathcal{P} } &= \sum_{m=1}^{M}  \sum_{k=1}^{K} \frac{q_{m} p_{mk}z_k}{W(\beta_k) p_{mk} z_k+V(\beta_k) \sum_{i=1}^{K}z_i}\nn \\
&= \sum_{k=1}^{K} \underbrace{\sum_{m=1}^M \frac{q_m  p_{mk} }{W(\beta_k) p_{mk}
+\widetilde{V}(\beta_k)}}_{\widetilde{\mathcal{P}}_k},
\end{align}
where $\widetilde{V}(\beta_k)=V(\beta_k) \frac{\sum_{i=1}^{K}z_i}{z_k}$. Thus, $\widetilde{\mathcal{P} } = \sum_{k=1}^{K} \widetilde{\mathcal{P}}_k$ where $\{\widetilde{\mathcal{P}}_k\}$ are independent of  each other. As a result, maximizing $\widetilde{\mathcal{P} } $ is equivalent to  separate maximization of individual  summation terms  $\{\widetilde{\mathcal{P}}_k\}$.  Therefore, Problem P0  can be approximated by $K$   single-tier optimization problems,  each of which  is written as:
\begin{equation}\tag{$\textbf{P4}$}\label{Opt_GenralCase}
\begin{aligned}
\mathop {\max}\limits_{\bf{p_k}} \; &{\sum_{m=1}^{M} q_{m} \frac{p_{mk}}{W(\beta_k) p_{mk} + \widetilde{V}(\beta_k)} } \\
{\textmd{s.t.}}\;\;&{\sum_{m=1}^{M}p_{mk}\leq C_k},  \\
& p_{mk}\in [0,1],  \forall m.
\end{aligned}
\end{equation}
Using the results in the case of single-tier HCNs in Theorem
\ref{OptimalCaching_SingleTier}, we derive the sub-optimal content-placement policy as shown in the following proposition. 
\begin{proposition}[Sub-Optimal TLCP for Multi-Tier HCNs with Non-Uniform SIRs]\label{Sub-OptimalCaching_MultiTier}
\emph{For the multi-tier cache-enabled HCNs with non-uniform SIR thresholds,  the optimal
TLCP placement probabilities, denoted by $\{\tilde{p}^*_{mk}\}$, that solve Problem~P4  are given as
\begin{equation}
\tilde{p}_{mk}^{*}(u_k^*)=\begin{cases}
1, & \text{if} \;\; q_m \geq \widetilde{T}_{1k}, \\
\frac{\sqrt{\widetilde{V}(\beta_k)}}{\sqrt{u_k^*}W(\beta_k)}\sqrt{q_m} - \frac{\widetilde{V}(\beta_k)}{W(\beta_k)}, &
\text{if}\;\;  \widetilde{T}_{0k} < q_m < \widetilde{T}_{1k},\!\! \\
0, & \text{if} \;\; q_m \leq \widetilde{T}_{0k},
\end{cases}
\end{equation}
where $\widetilde{T}_{1k}=\frac{u_k^* (W(\beta_k)+\widetilde{V}(\beta_k))^2}{\widetilde{V}(\beta_k)}$ and $\widetilde{T}_{0k}=u_k^* \widetilde{V}(\beta_k)$.
The optimal dual variable $u_k^*$ satisfies the equality
\begin{equation}
\sum_{m=1}^M p_{mk}^*(u_k^*)=C_k.
\end{equation}
}
\end{proposition}
The above sub-optimal TLCP policy approximates problem P0 as $K$ independent single-tier optimization problems. Thus, the corresponding computational complexity is $O(DMK)$, where $D$ denotes the number of iterations needed to find the optimal Lagrange multiplier for each tier. In addition, the numerical results  in the next section show that it  can attain close-to-optimal performance.

\section{Simulation Results}
In this section, simulation is conducted to validate the optimality of the  content-placement
policies derived in the preceding section and to compare the performance of the strategy of TLCP with conventional ones. The benchmark strategies include the
\emph{``most popular'' content placement} (MPCP) that caches the most popular contents in a
greedy manner and the \emph{hybrid content placement} (HCP) proposed in
\cite{CacheMulticast_HWN}. Our simulation is based on the following settings unless
specified otherwise. The number of BS tiers is $K=2$ and  the path-loss exponent $\alpha=3$.
The BS transmission power for the two tiers are $P_1=46$ dBm and $P_2=30$ dBm, respectively.
The SIR threshold for tier~$1$ is fixed at $\beta_1=-4$ dB while the other $\beta_2$ is a variable. 

\begin{figure}
\centering \includegraphics[width=8.5cm]{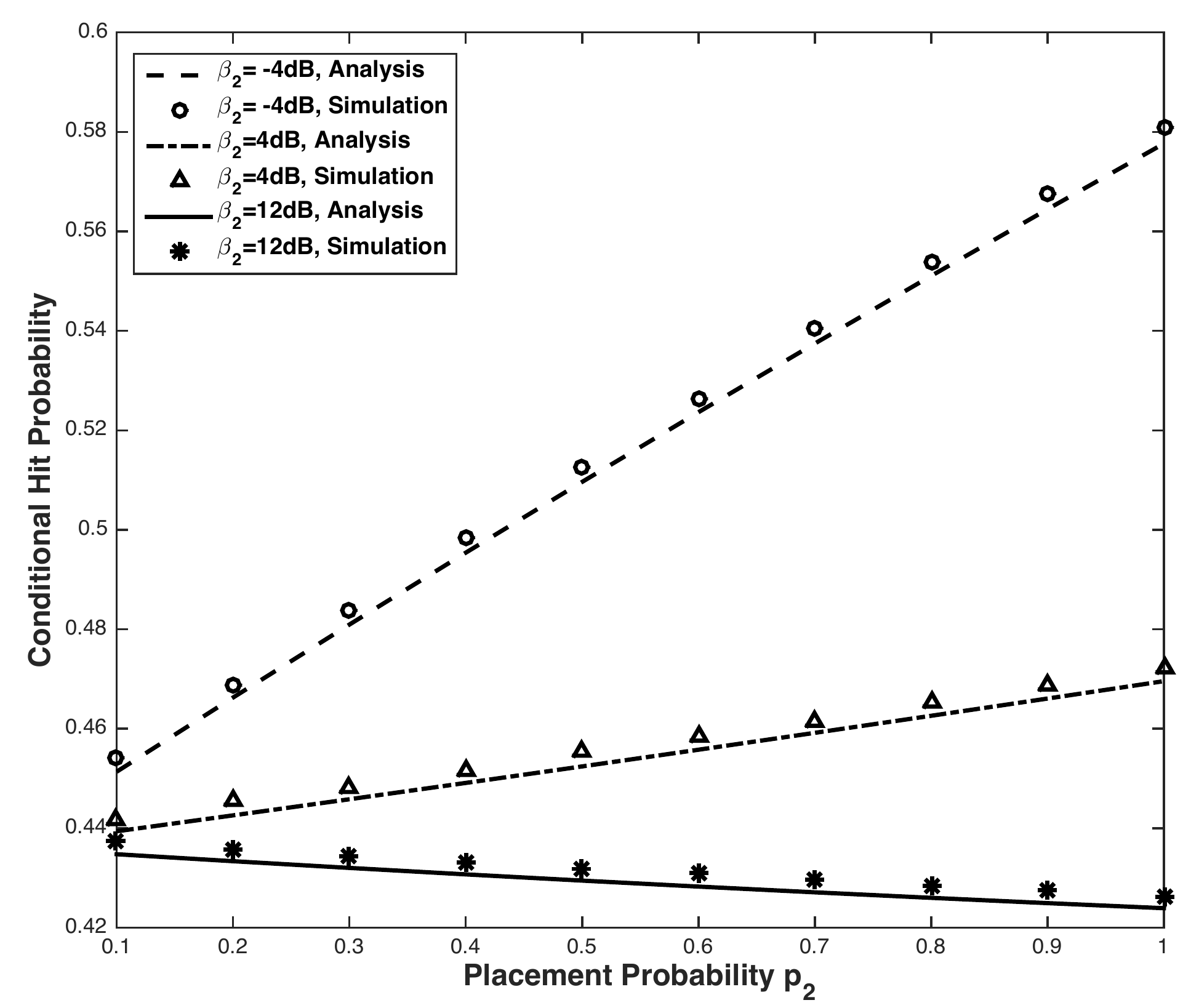}
\caption{Conditional hit probability versus caching probability with $\lambda_2=5 \lambda_1$, $p_1=1$.}\label{Fig:ConHitProb}
\end{figure}

\begin{figure*}[t]
\centering
\subfigure[Optimal hit probabilities under different content placement policies with $\beta_1=\beta_2=-4$ dB, $C_1=10$, $C_2=8$, $M=20$.]{\includegraphics[width=8cm]{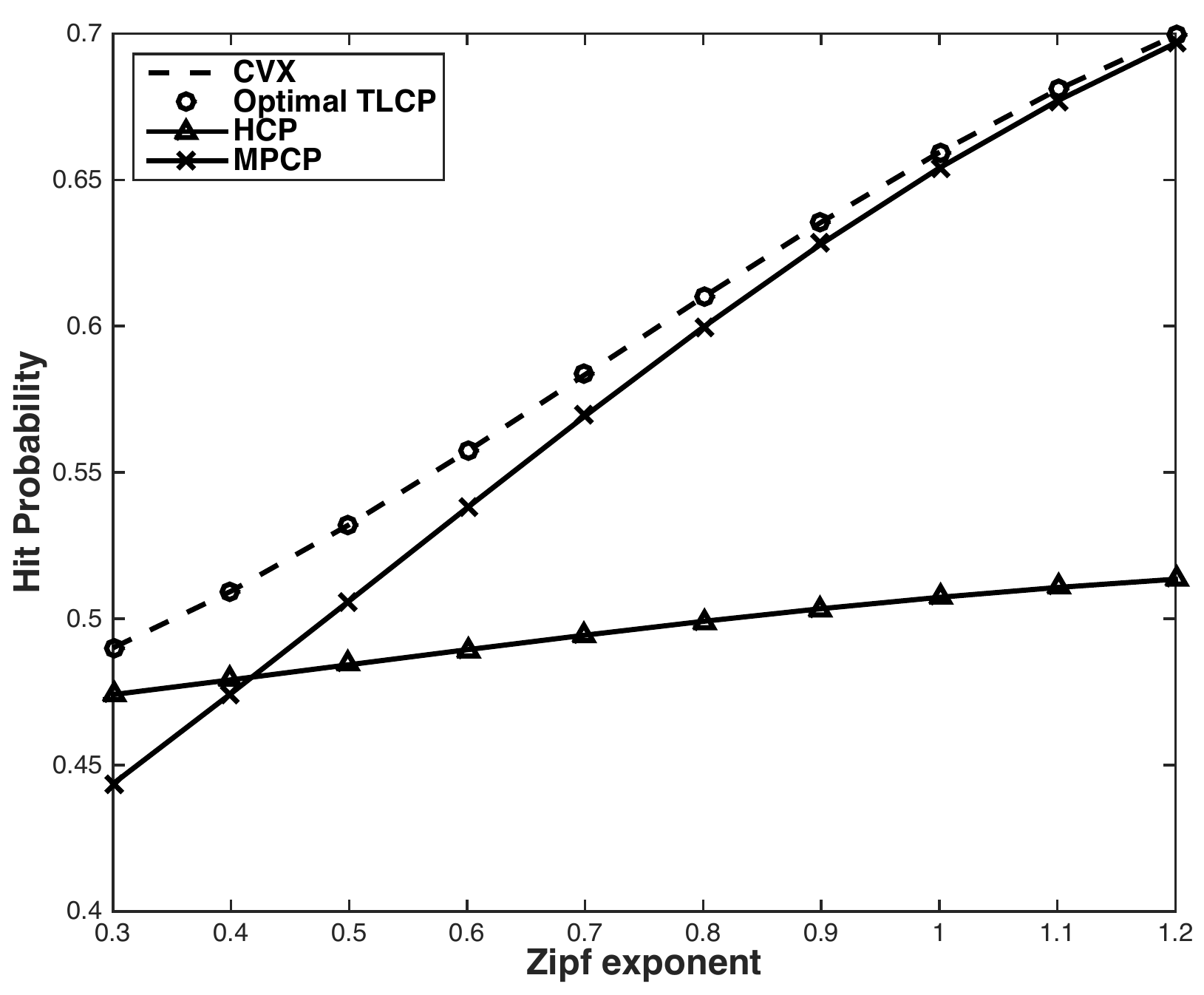}}
\subfigure[Optimal Hit probability for different throughput with $C_1=60$, $C_2=40$, $M=120$]{\includegraphics[width=8cm]{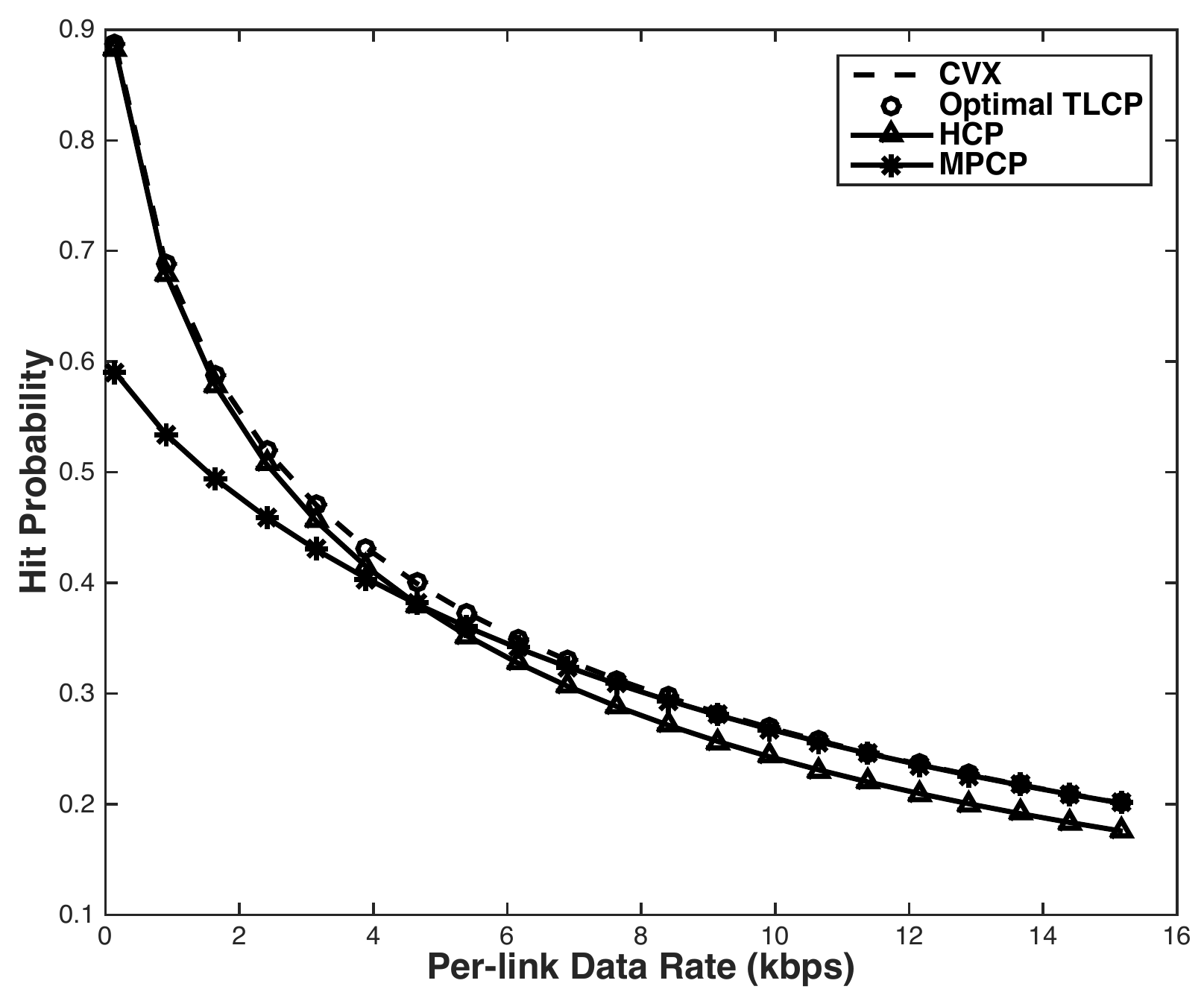}}
\caption{Optimal hit probabilities under different content placement policies with $\lambda_2=10\lambda_1$.}\label{Fig:OptimalHitProb_Compare}
\end{figure*}

\subsection{Conditional Hit Probability}
The conditional hit probability for a typical file $\mathcal{F}_{m}$ versus caching probability $p_2$ is shown in Fig. \ref{Fig:ConHitProb}. The analytical results are computed numerically using Lemma
\ref{Lem:ConditionalHitProb} and the simulated ones are obtained from Monte Carlo simulation using Matlab. First, it is observed that the simulated results match the analytical results well,
which validates our analysis. In addition, the conditional hit probability
increases with the growing placement probability $p_2$ if $\beta_2=\beta_1$. However, it
does not necessarily hold for the case $\beta_2>\beta_1$, which shows that the effects of
placement probability on the hit probability differ with SIR threshold. This is because
increasing the placement probability increases the association probability of that tier (see
Lemma \ref{Lemma:Association probability}) and thus decreases the conditional hit
probability if that tier has smaller hit probability due to the larger SIR threshold.
Meanwhile, it reduces the serving distance (see Lemma \ref{Lemma: Serving distance}) and
thus increases the conditional hit probability. The (final) effects of placement probability
on the hit probability are determined by the absolute values of the above increment and decrement. 

\subsection{Optimal Content Placement}

\begin{figure*}[t]
\centering
\subfigure[Hit probability for different $C_1$ with $C_2=8$, $M=20$.]{\includegraphics[width=8cm]{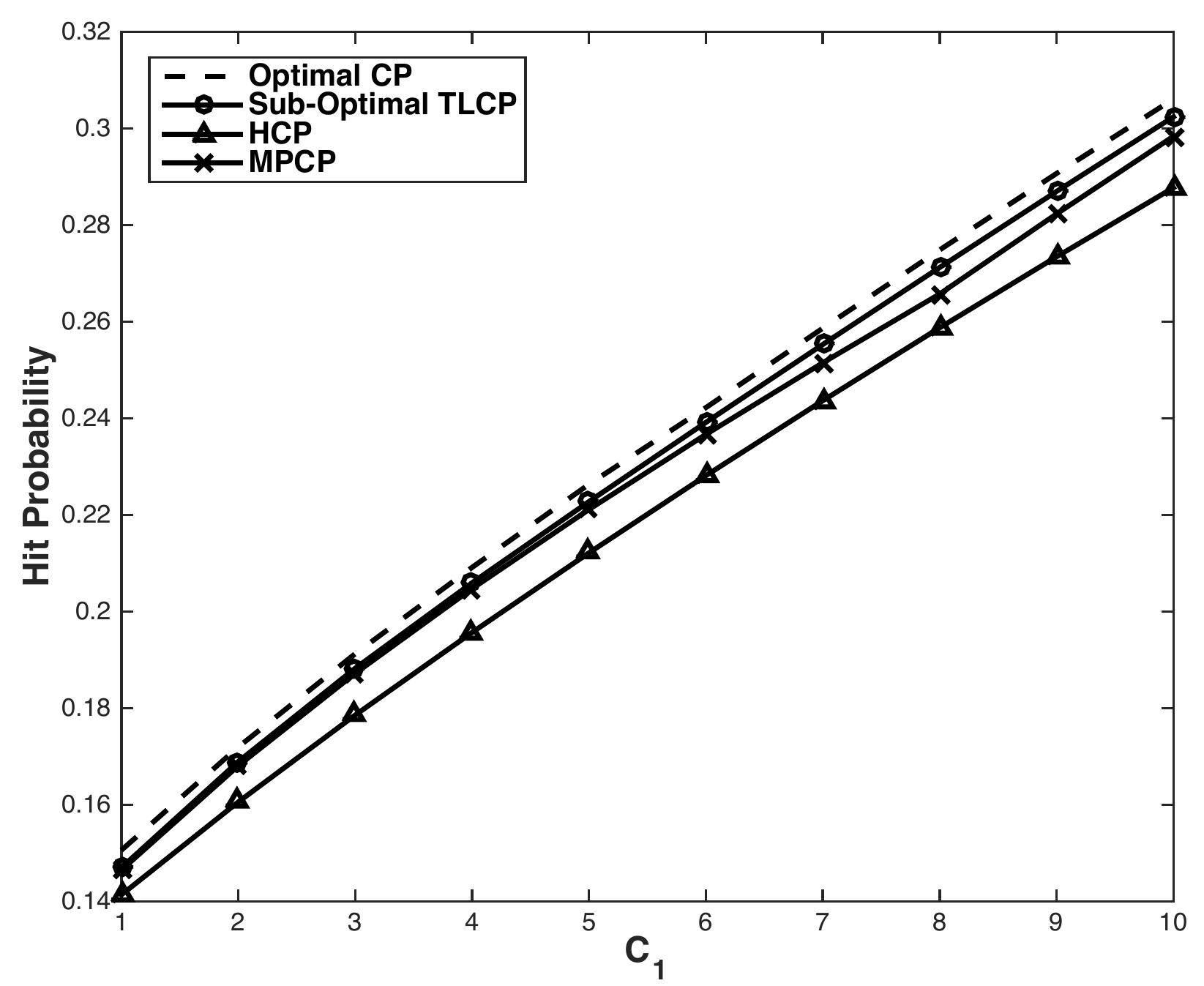}}
\subfigure[Hit probability for different $M$ with $C_1=60$, $C_2=40$]{\includegraphics[width=8cm]{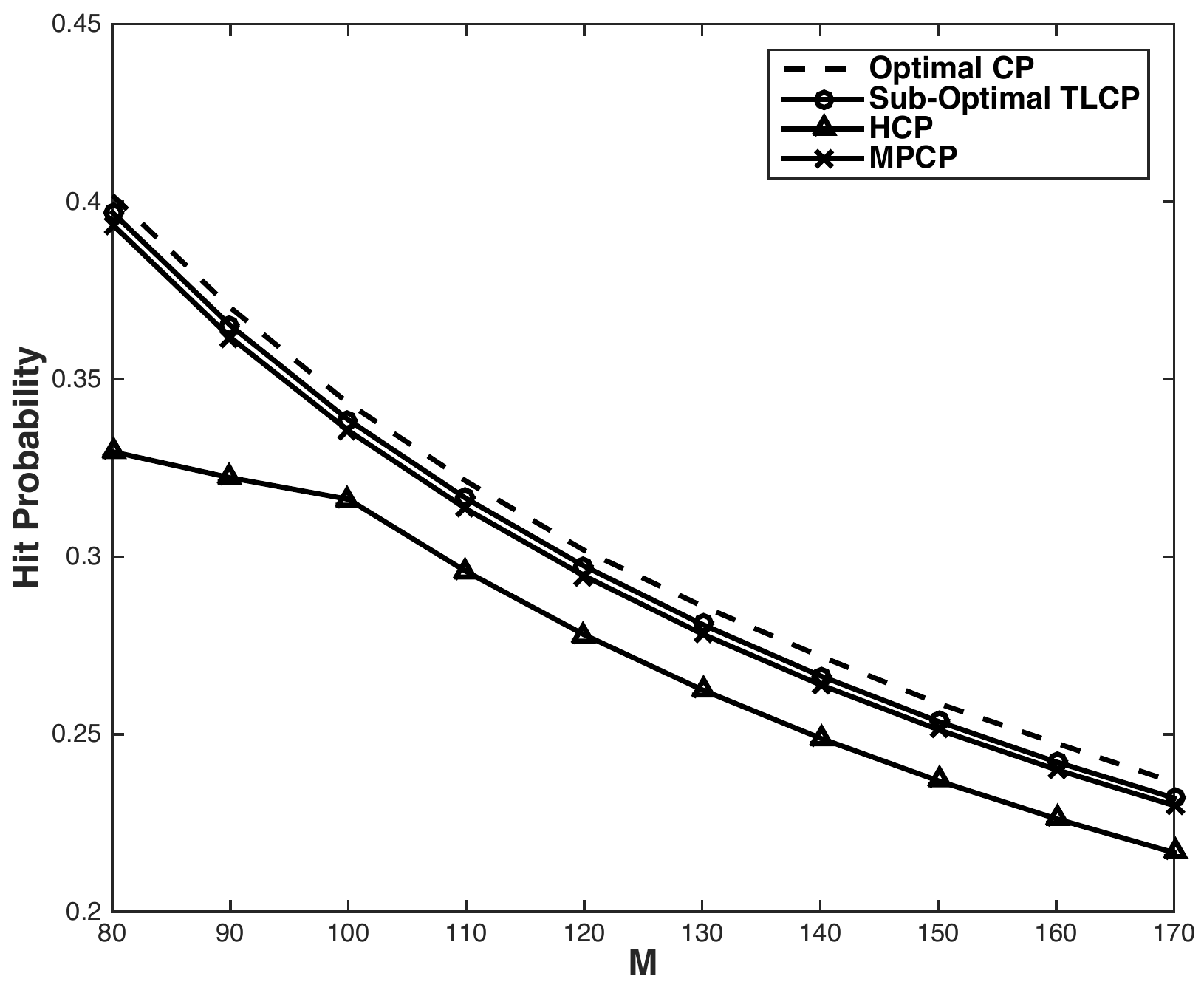}}
\caption{Hit probabilities under different content placement policies with $\lambda_2=10\lambda_1$, $\beta_1=-4$ dB, $\beta_2=-2$ dB.}\label{Fig:OptimalHitProb}
\end{figure*}

Fig. \ref{Fig:OptimalHitProb_Compare} compares the performance of the optimal TLCP proposed in this
paper (Theorem \ref{Tm:OptimalSetting}) with MPCP and
HCP. For MPCP, each macro-cell BS (or small-cell BS) caches the
$C_1$ (or $C_2$) most popular files. For HCP,  each macro-cell BS caches the $C_1$ most
popular files while each small-cell BS caches the remaining files with optimal probabilistic
content placement given in  Theorem \ref{OptimalCaching_SingleTier}. First of all, Fig. \ref{Fig:OptimalHitProb_Compare} (a) shows that the hit probabilities under these three content placement policies increase as
the content popularity becomes more skewed (a growing Zipf exponent $\gamma$),  aligned with
intuition. Next, TLCP is observed to achieve higher hit probability than MPCP and HCP due to
the content densification and diversity (see Remark \ref{MultiplexityDiversity}). Further, we
observe that the gain over MPCP decreases with a growing $\gamma$ since MPCP is a 
popularity-aware policy. In contrast, the gain over HCP increases with a growing $\gamma$. This is because, in the HCP, only Macro-cell tier caches the $C_1$ most popular files. In addition, the optimality of TLCP is verified by comparing the results given by the standard optimization tool CVX. Last, from Fig. \ref{Fig:OptimalHitProb_Compare} (b), it is observed that the optimal hit probability increases as the per-link data rate reduces due to the reducing SIR threshold. In particular, the maximum hit probability (when the data rate approximates to 0) is less than 1 since the caching capacity is limited.

The hit probabilities under different CP policies, including the optimal CP, sub-optimal TLCP (see Proposition
\ref{Sub-OptimalCaching_MultiTier}), MPCP, and HCP, versus different cache capacities and the number of contents are shown
in Fig. \ref{Fig:OptimalHitProb}(a) and Fig. \ref{Fig:OptimalHitProb}(b), respectively. The optimal CP under this case (i.e., multi-tier HCNs with non-uniform SIR thresholds) is derived by adopting the dual methods for non-convex optimization problem in \cite{TimeSharing} since Problem P0 has the same structure as that in \cite{TimeSharing} and it satisfies the time-sharing condition (the proof is shown in Appendix \ref{Pf:Time-sharing}). Compared with the optimal CP, the sub-optimal TLCP provides close-to-optimal performance. It should be noted that the computational complexity of optimal solution is linear in the number of files $M$, but exponential in the number of BS tiers $K$, since it involves solving $M$ nonconvex optimization problems, corresponding to the $M$ tones, each with $K$ variables. While the computational complexity of our proposed sub-optimal TLCP algorithm is linear with \emph{both} $M$ and $K$. In addition, besides the obvious monotone-increasing
hit probability with cache capacity, we observe that the sub-optimal TLCP outperforms both the HCP and MPCP. Finally, it is shown that the hit probability increases with the growing cache capacity and decreases with the growing number of contents, which coincides with our intuition.

\section{Conclusion}
In this paper, we have studied the hit probability and the optimal content placement of the
cache-enabled HCNs where the BSs are distributed as multiple independent PPPs and the files
are probabilistically cached at BSs in different tiers with different BS densities,
transmission powers, cache capacities and SIR thresholds. Using stochastic geometry, we have
analyzed the hit probability and shown that it is affected by both the physical layer and
content-related parameters. Specifically, for the case where all the tiers have the uniform SIR thresholds,
the hit probability increases with all the placement probabilities and converges to its
maximum (constant) value as all the probabilities achieve one without considering the cache
capacity constraint. Then, with the cache capacity constraint, the optimal
content placement strategy has been proposed to maximize the hit probability for both single- and
multi-tier HCNs. We have found that the placement probability for each file has the OPP
caching structure, i.e., the optimal placement probability is linearly proportional to the
square root of offset-popularity with truncation to enforce the range for the probability.
On the other hand, for multi-tier HCNs with uniform SIR thresholds, interestingly, the
weighted-sum of the optimal placement probabilities also has the OPP caching structure. Further, an optimal or a sub-optimal TLCP caching algorithm has been proposed to maximize the hit probability HCNs with uniform or non-uniform SIR thresholds, respectively.

The fundamental structure of the optimal content placement strategies proposed in this paper
provides useful guidelines and insights for designing cache-enabled wireless networks. As
a promising future direction, it would be very helpful to take BS cooperation and multicast
transmissions into account for practical networks. In addition, coded caching can be used to further enhance network performance.

\appendix{}

\subsection{Analysis of Backhaul Latency}\label{App:P_Delay}

Based on (9) in  \cite{Quek_BackhaulNetworks}, the mean packet delay in propagation via a wired backhaul network can be  approximated as 
\begin{equation} \label{Eq:T_uc}
T_{uc}=\left(1+1.28\frac{\lambda_b}{\lambda_g}\right)c_1+c_2,
\end{equation}
where $\lambda_b$ denotes the BS density, $\lambda_g$ is the gateway density, $c_1$ and $c_2$ are constants related to  the processing capability of a backhaul node.  In cache-enabled HCNs, the typical user has to retrieve using the backhaul network a  file that  is not cached at BSs. It follows from \eqref{Eq:T_uc} that the resultant  backhaul latency is given as
\begin{equation}\label{Eq:T_c}
T_c = (1-\mathcal{P})
\l(1+1.28\frac{(1-\mathcal{P})\lambda_b}{\lambda_g}\r)c_1+c_2,
\end{equation}
where $\mathcal{P}$ is the hit probability for cache-enabled HCNs. From (\ref{Eq:T_c}), we can see that the backhaul latency is a monotone-decreasing function of the hit probability as shown in Fig. \ref{Fig:P_Delay}. This shows that improving the hit probability of the radio access network reduces the burden on the backhaul network.

\begin{figure}
\centering \includegraphics[width=8.5cm]{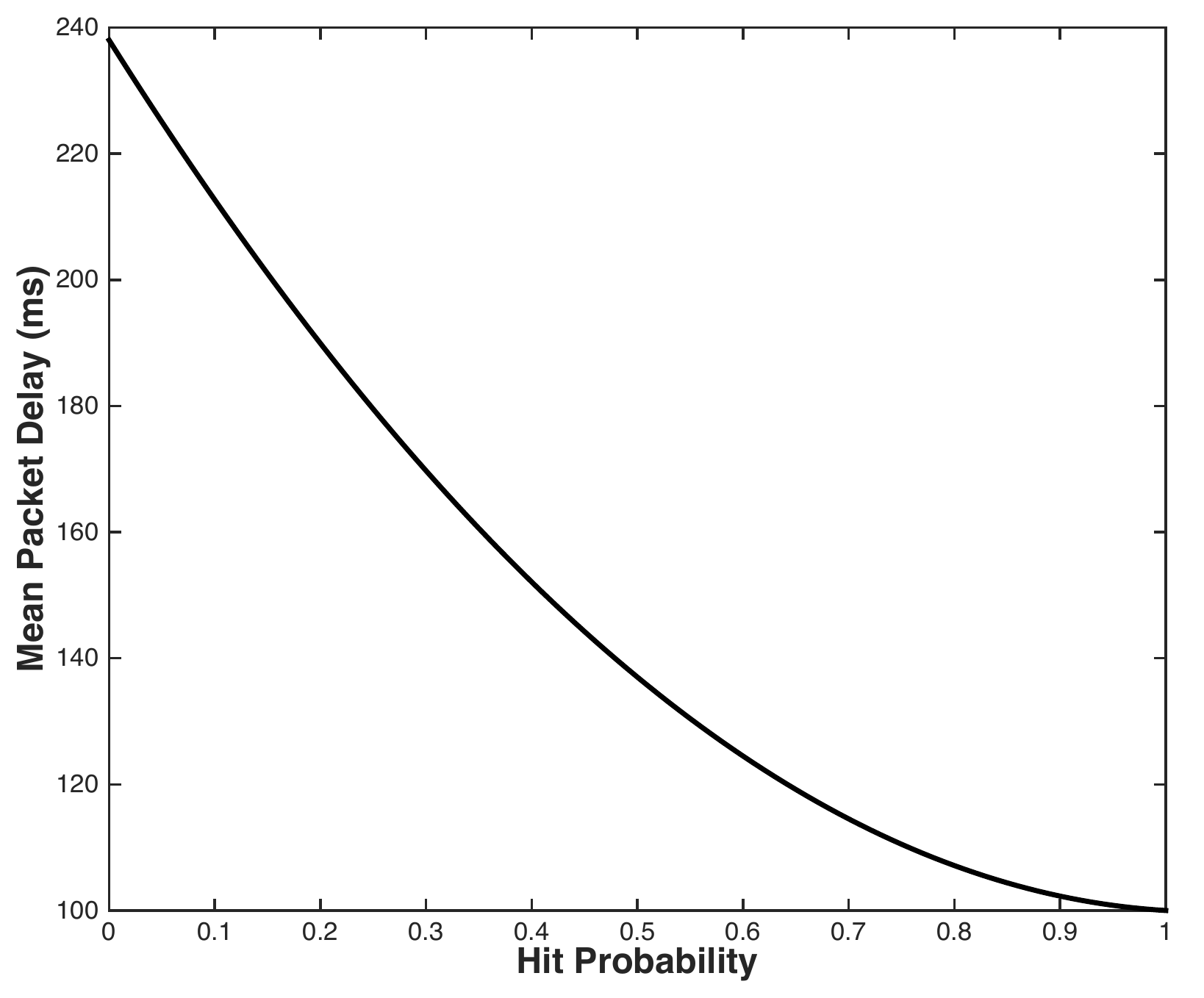}
\caption{Mean packet delay in wired backhaul versus hit probability. Here, $\lambda_b=10\lambda_g$, $c_1=10$ms, and $c_2=100$ms.}\label{Fig:P_Delay}
\end{figure}

\subsection{Proof of Lemma 1}\label{Pf: Association-Probability}

Define  $P_{\mathrm{r},m,k}=P_{k}R_{mk}^{-\alpha}$, which represents  the received signal  power due to transmissions by the BSs with
file $\mathcal{F}_m$ in the $k$-th tier, where $R_{mk}$ is the distance from the typical user to the nearest BS in content-centric tier $\Phi_{mk}$. According to the content-centric cell association, the association probability $A_{mk}$ is the probability that $P_{\mathrm{r},m,k}>P_{\mathrm{r},m,j},\:\forall j,\: j\neq k$.
Therefore,
\begin{align}
 A_{mk} &=\mathbb{E}_{R_{mk}}\left[\mathbb{P}\left[P_{\mathrm{r},m,k}\left(R_{mk}\right)>\max_{j,j\neq k}P_{\mathrm{r},m,j}\right]\right]\nonumber \\
 & =\mathbb{E}_{R_{mk}}\left[\prod_{j=1,j\neq k}^{K}\mathbb{P}\left[P_{\mathrm{r},m,k}\left(R_{mk}\right)>P_{\mathrm{r},m,j}\right]\right]\nonumber \\
 & =\mathbb{E}_{R_{mk}}\!\left[\prod_{j=1,j\neq k}^{K}\!\!\mathbb{P}\left[R_{mj}\!>\!(P_{j}/P_{k})^{1/\alpha }R_{mk}\right]\right] \nonumber \\
 & =\!\!\int_{0}^{\infty}\!\!\!\!\prod_{j=1,j\neq k}^{K}\!\!\mathbb{P}\left[R_{mj}\!>\!(P_{j}/P_{k})^{1/\alpha } r\right]f_{R_{mk}}(r)\mathrm{d}r.\label{eq:Amk_org}
\end{align}
To derive $A_{mk}$, $\mathbb{P}\left[R_{mj}>(P_{j}/P_{k})^{1/\alpha }r\right]$ and the
probability density function (PDF) of $R_{mk}$, denoted by $f_{R_{mk}}(r)$, are calculated as follows.
\begin{align}
 & \mathbb{P}\left[R_{mj}>(P_{j}/P_{k})^{1/\alpha } r\right]\nonumber \\
= & \mathbb{P}\left[\mathrm{No}\:\mathrm{BS}\:\mathrm{with}\:\mathrm{file}\: f_{m}\:\mathrm{closer}\:\mathrm{than}\:\left((P_{j}/P_{k})^{1/\alpha }r\right)\:\mathrm{in}\:\mathrm{the}\: j\mathrm{th}\:\mathrm{tier}\right]\nonumber \\
= & \exp\left(-\pi p_{mj} \lambda_{j}(P_{j}/P_{k})^{2/\alpha }r^{2}\right).\label{eq:Amk_step1}
\end{align}
Further, $f_{R_{mk}}(r)$ is derived by taking the derivative of $1-\mathbb{P}\left[R_{mk}>r\right]$ with respect
to $r$,
\begin{flalign}
f_{R_{mk}}\left(r\right) & =\frac{\mathrm{d}\left[1-\mathbb{P}\left[R_{mk}>r\right]\right]}{\mathrm{d}r}=2\pi  p_{mk} \lambda_{k}r\exp\left(-\pi p_{mk} \lambda_{k}r^{2}\right).\label{eq:Amk_step2}
\end{flalign}
Last, the expression of $A_{mk}$ is derived by substituting (\ref{eq:Amk_step1}) and (\ref{eq:Amk_step2}) into (\ref{eq:Amk_org}).

\subsection{Proof of Lemma \ref{Lem:ConditionalHitProb}}\label{Pf:ConditionalHitProb}
In order to calculate the conditional hit probability, we first derive the probability that
the typical user successfully receives the requested file from its given serving BS $X_{k}$
in the $k$-th tier, denoted by $\mathcal{P}_{m}(X_{k},k)$, as follows.
\begin{align}
& \mathcal{P}_{m}(X_{k},k) =\mathbb{P}\left[\frac{P_{k}h_{X_{k}}||X_{k}||^{-\alpha}}{{I}(\mathcal{F}_m)}>\beta_k \right]\nonumber \\
 & \overset{\left(a\right)}{=}\mathbb{E}_{{I}(\mathcal{F}_m)}\left[\exp\left(-\frac{\beta_k {I}(\mathcal{F}_m)||X_{k}||^{\alpha}}{P_{k}}\right)\right]\nonumber \\
 & \overset{\left(b\right)}{=}\prod_{i=1}^{K}\mathbb{E}_{\Phi_{mi}}\!\!\left\{ \prod_{X \in \Phi_{mi} \setminus X_{k}} \!\! \mathbb{E}_{h_X}\!\! \left[ \exp\left(-\frac{\beta_k P_i h_X||X_{k}||^{\alpha}}{P_{k} ||X||^{\alpha}}\right) \right] \right\} \nonumber \\
 & \;\;\;\;\cdot \prod_{i=1}^{K}\mathbb{E}_{\Phi_{mi}^{c}}\!\! \left\{ \prod_{X \in \Phi_{mi}^{c}} \!\! \mathbb{E}_{h_X}\!\! \left[ \exp\left(-\frac{\beta_k P_i h_X||X_{k}||^{\alpha}}{P_{k} ||X||^{\alpha}}\right) \right] \right\} \nonumber \\
 & \overset{\left(c\right)}{=}\prod_{i=1}^{K}\mathbb{E}_{\Phi_{mi}}\left[ \prod_{X \in \Phi_{mi} \setminus X_{k}} \left( 1+\frac{\beta_k P_i ||X_{k}||^{\alpha}}{P_{k} ||X||^{\alpha}}\right)^{-1} \right] \nonumber \\
 & \;\;\;\;\cdot \prod_{i=1}^{K}\mathbb{E}_{\Phi_{mi}^{c}}\left[ \prod_{X \in \Phi_{mi}^{c}} \left( 1+\frac{\beta_k P_i ||X_{k}||^{\alpha}}{P_{k} ||X||^{\alpha}}\right)^{-1} \right] \nonumber \\
 &  \overset{\left(d\right)}{=}\prod_{i=1}^{K} \exp \! \left\{ -\! \int_{\mathbb{R}^2\setminus\mathrm{b}(o, z_i)} \! \left[ 1\!-\!\left( 1\!+\! \frac{\theta}{||X||^{\alpha}}\right)^{-1}\! \right] p_{mi} \lambda_i \mathrm{d}X \! \right\} \nonumber\\
 & \;\;\;\;\cdot \prod_{i=1}^{K} \exp \! \left\{ -\! \int_{\mathbb{R}^2} \!\left[ 1\!-\!\left( 1\!+\! \frac{\theta}{||X||^{\alpha}}\right)^{-1}\!\right] (1-p_{mi}) \lambda_i \mathrm{d}X \!\right\} \nonumber\\
 &  \overset{\left(e\right)}{=}\prod_{i=1}^{K} \exp \left\{ -2 \pi p_{mi} \lambda_i \int_{z_i}^{\infty} \left[ 1-\left( 1+ \frac{\theta}{r^{\alpha}}\right)^{-1}\right] r \mathrm{d}r \right\} \nonumber\\
 & \;\;\;\;\cdot \prod_{i=1}^{K} \exp \! \left\{ -2 \pi (1\!-\!p_{mi}) \lambda_i \! \int_{0}^{\infty} \left[ 1\!-\!\left( 1\!+\! \frac{\theta}{r^{\alpha}}\right)^{-1} \! \right] r \mathrm{d}r \! \right\} \nonumber\\
 & =\prod_{i=1}^{K} \exp \left\{ -2 \pi p_{mi} \lambda_i \int_{z_i}^{\infty} \frac{r}{1+\frac{r^{\alpha}}{\theta}} \mathrm{d}r\right\} \nonumber \\
 & \;\;\;\;\cdot \prod_{i=1}^{K} \exp \left\{ -2 \pi (1-p_{mi}) \lambda_i \int_{0}^{\infty} \frac{r}{1+\frac{r^{\alpha}}{\theta}} \mathrm{d}r\right\} \nonumber \\
 & \overset{\left(f\right)}{=}\prod_{i=1}^{K} \exp \!\!\left[ - \delta \pi p_{mi} \lambda_i
 \frac{\theta z_{i}^{\alpha(\delta\!-\!1)}}{1\!-\!\delta}
 {}_{2}F_{1}\left(1,1\!-\!\delta;2\!-\!\delta;-\frac{\theta}{z_i^{\alpha}}\right)\right] \nonumber \\
 & \;\;\;\;\cdot \prod_{i=1}^{K} \exp\!\!\left[ -\delta \pi (1\!-\!p_{mi}) \lambda_i \theta ^{\delta} B(\delta,1\!-\!\delta)\!\right]\nonumber \\
 & \overset{\left(g\right)}{=}\exp \!\!\left[ -\!\sum_{i=1}^{K} \pi p_{mi} \lambda_i \!\left( \frac{P_i}{P_k}\right)^{\delta} \!\! Q\left(\!\beta_k \!\right) ||X_{k}||^2 \!\right] \nonumber \\
 & \;\;\;\;\cdot \!\exp \!\left[ -\!\sum_{i=1}^{K} V(\beta_k) \pi (1\!-\!p_{mi}) \lambda_i
 \!\left( \frac{P_i}{P_k}\right)^{\delta} \!\! ||X_{k}||^2 \right],
  \label{eq:p_Xmk}
\end{align}
where $\left(a\right)$ and $\left(c\right)$ come from taking expectation with respect to
$h\sim\exp\left(1\right)$; $\left(b\right)$ follows from the expression of
${I}(\mathcal{F}_m)$ in \eqref{Eq:I:Cond}; $\left(d\right)$ follows from the probability
generating functional of PPP, 
$\theta=\beta_k P_i ||X_{k}||^{\alpha}/P_k$, and $\mathrm{b}(o, z_i)$ denotes a ball of radius $z_i$ centered at the origin denoted by $o$. Note that the closest interferer in $\Phi_{mi}$ is at least at the distance
$z_i=(P_i/P_k)^{1/\alpha} ||X_{k}||$, while all the BSs in $\Phi_{mi}^{c}$ are the interferers. Equality (e) comes from converting from Cartesian to polar coordinates, (f) follows
from replacing $r^{\alpha}$ with $u$ and calculating the corresponding integral based on the
formula (3.194.2) and (3.194.3) in \cite{IntegralTable}, and finally (g) comes from the expressions of $z_i$, $\theta$, $Q(\beta_k)$, and $V(\beta_k)$.

Averaging over the distance $||X_{k}||$, we have the hit probability for a user requesting for
file $\mathcal{F}_{m}$ in the $k$-th tier:
\begin{align}
\mathcal{P}_{m}(k) \!=\!\mathbb{E}_{||X_{k}||}\left[\mathcal{P}_{m}(X_{k},k)\right]\!=\!\int_{0}^{\infty}\!\!\!\mathcal{P}_{m}(X_{k},k) f_{||X_{k}||}(x)\mathrm{d}x.\label{eq:p_k}
\end{align}
Last, the result of Lemma \ref{Lem:ConditionalHitProb} is obtained by substituting (\ref{eq:AssocProb})
and (\ref{eq:p_k}) into \eqref{Eq:HP_File_def} based on the law of total probability.

\subsection{Proof of Proposition \ref{Prop_WeightedSum}} \label{Pf:OptimalCachingProb_IdenticalBeta}

Since Problem P2 is convex, it can be solved by using the Lagrange method. The corresponding partial Lagrangian function is
\begin{equation}
L(\mathbf{p},\mathbf{u})\!=\!\sum_{m=1}^{M} q_{m} \frac{\sum_{k=1}^{K}p_{mk}z_k}{W \sum_{k=1}^{K}p_{mk}z_k + V^{\prime} }+\sum_{k=1}^{K}\!u_{k}(C_{k}\!-\!\sum_{m=1}^{M}\!p_{mk}),
\end{equation}
where $\mathbf{u}\!=\!(u_1,\ldots,u_K)\!\geq \!0$ denotes the Lagrangian multiplier. Taking derivative of $L(\mathbf{p},\mathbf{u})$ with respect to $p_{mk}$, we have
\begin{equation}
\frac{\partial L}{\partial p_{mk}} = \frac{q_m V^{\prime} z_{k}}{(V^{\prime}+W
\sum_{k=1}^{K} p_{mk} z_{k})^{2}} - u_{k}.
\end{equation}
Thus, given the optimal Lagrangian multiplier $\mathbf{u}^*$, the optimal placement probability $p_{mk}^{*}$ is expressed as
\begin{equation} \label{Eq:p_mk}
p_{mk}^{*}(u_{k}^*)=\begin{cases}
  1, & \text{if} \;\;\;\; q_m \geq \frac{u_k^* (W \sum_{k=1}^{K}
  z_k+V^{\prime})^2}{V^{\prime} z_{k}},\\
  \xi(u_k^*), & \text{if}\;\;\;\;  \frac{u_k^* V ^{\prime}}{z_k}< q_m < \frac{u_k^* (W
  \sum_{k=1}^{K} z_k+V^{\prime})^2}{V^{\prime} z_{k}},\\
  0, & \text{if} \;\;\;\; q_m \leq \frac{u_k^* V ^{\prime}}{z_k},
\end{cases}
\end{equation}
where $\xi(u_k^*)$ is the solution over $p_{mk}$ of the equation
\begin{equation}
\frac{q_m V^{\prime} z_{k}}{(V^{\prime}+W \sum_{k=1}^{K} p_{mk}^{*} z_{k})^{2}} = u_{k}^{*}. 
\end{equation}
Thus, we have
\begin{equation}\label{Eq:Xi}
\sum_{k=1}^{K} p_{mk}^{*} z_{k} = \left(\sqrt{q_m V^{\prime} z_k/u_k^*}-V^{\prime}\right)/W. 
\end{equation}
Note that \eqref{Eq:Xi} holds for all $k$. Thus, $\xi(u_k^*)$ in Eq.~(\ref{Eq:p_mk}) satisfies the following equation by denoting $\eta^*\!\!=\!\!u_k^*/z_k$ 
\begin{equation} \label{Eq:sum}
\sum_{k=1}^{K} p_{mk}^{*} z_{k} =\left(\sqrt{q_m V^{\prime}/\eta^*}-V^{\prime}\right)/W. 
\end{equation}

According to the KKT conditions, the dual variable $u_{k}$ satisfies the following equation:
\begin{equation}
u_k \left(C_k-\sum_{m=1}^{M} p_{mk}^*\right)=0. 
\end{equation}
Thus, the alternative multiplier $\eta^*$ satisfies
\begin{equation}
\eta^* \left(\sum_{k=1}^{K}C_k z_k-\sum_{m=1}^{M} \sum_{k=1}^{K} p_{mk}^* z_k\right)=0.
\end{equation}
If $\eta^*=0$ ($u_k^*=0$), according to Eq. (\ref{Eq:p_mk}), all files should be cached with probability 1 which conflicts with our assumption of limited cache capacity. Thus, we have
\begin{equation}\label{Eq: MultiplierE}
\sum_{k=1}^{K}C_k z_k=\sum_{m=1}^{M} \sum_{k=1}^{K} p_{mk}^* z_k. 
\end{equation}
According to  Eq. (\ref{Eq:p_mk}), Eq. (\ref{Eq:sum}) and Eq. (\ref{Eq: MultiplierE}), we
have $\sum_{k=1}^{K} p_{mk}^* z_k=g_{m}^*,  m = 1, 2, \ldots, M$.

\subsection{Proof of Theorem \ref{Tm:OptimalSetting}}\label{Pf:OptimalSetting}

In order to prove that ${p}_{mk}^*$ given in Theorem \ref{Tm:OptimalSetting} is the optimal placement probability for Problem P2, we need to follow the following  two steps: (1) $\sum_{k=1}^{K}{p}_{mk}^* z_k=g_m^*$ where $g_m^*$ is given in Proposition \ref{Prop:SumWeighted}. (2) $\sum_{m=1}^{M} p_{mk}^*=C_k$. This is because the objective function of Problem P2 monotonically increases with the growing placement probability and thus the optimal content probabilities satisfy the relaxed constraint with equality.

(1) Proof $\sum_{k=1}^{K}{p}_{mk}^* z_k=g_m^*$ :

To this end, we first prove $\sum_{k=1}^{K} \zeta_{mj}^* z_k=g_{m}^{*}$.
When $q_m \in (T_0^{\prime}, T_1^{\prime})$, we have $\sum_{i=m}^{M} \sum_{k=1}^{K} p_{ik}^* z_k = \sum_{i=m}^M g_i^*$ (since $g_m^*=\sum_{k=1}^{K} p_{ik}^* z_k$). On the other hand, $\sum_{i=m}^{M} \sum_{k=1}^{K} p_{ik}^* z_k =\sum_{k=1}^{K} \sum_{i=m}^{M} p_{ik}^* z_k = \sum_{k=1}^{K} C_k^{\prime} z_k$. Thus, we have the following equation:
\begin{equation}\label{eq:RowColumn}
\sum_{i=m}^M g_{i}^* = \sum_{k=1}^{K} C_k^{\prime} z_k.
\end{equation}
Based on \eqref{eq:RowColumn}, we have
\begin{equation}
\sum_{k=1}^{K} \zeta_{mj}^* z_k = \sum_{k=1}^{K} \frac{c_k^{\prime} g_m^* z_k}{\sum_{i=m}^{M} g_i^*} = \frac{g_m^* \sum_{k=1}^{K} c_k^{\prime} z_k}{\sum_{i=m}^{M} g_i^*} =g_m^*.
\end{equation}
Further, according to the expression of ${p}_{mk}^*$ in \eqref{Eq:OptimalCacheProb_MultiTier}, we have
\begin{equation}
\sum_{k=1}^{K} {p}_{mk}^* z_k = \sum_{k=1}^{K} \zeta_{mj}^* z_k = g_m^*.
\end{equation}
Thus, $\sum_{k=1}^{K}{p}_{mk}^* z_k=g_m^*$.

(2) Proof $\sum_{m=1}^{M} p_{mk}^*=C_k$:
\begin{align}
\sum_{m=1}^{M}{p}_{mk}^* &= \sum_{m=1}^{M-1} {p}_{mk}^* +  {p}_{Mk}^* = \sum_{m=1}^{M-1} {p}_{mk}^* +  C_k^{\prime} \nonumber \\
&= \sum_{m=1}^{M-1}{p}_{mk}^* + C_k - \sum_{m=1}^{M-1}{p}_{mk}^* = C_k.
\end{align}

\subsection{Proof of Time-sharing Condition} \label{Pf:Time-sharing}
In order to prove the Problem (P0) in this paper can be solved by the dual methods for nonconvex optimization problem in \cite{TimeSharing}, the following two steps are needed: 

(1) Prove Problem (P0) has the same structure as (4) in \cite{TimeSharing}:

By observation, the optimization problem (P0) can be rewritten as follows:
\begin{equation}\label{Opt_Transfer}
\begin{aligned}
\mathop {\max} \;\; &{\sum_{m=1}^{M} f_{m}(\mathbf{P}_{m})}\\
{\textmd{s.t.}}\;\;&{\sum_{m=1}^{M} h_{m}(\mathbf{P}_{m}) \leq \mathbf{C}}\\
\end{aligned}
\end{equation}
where $\mathbf{P}_{m} = (p_{m1},\cdots, p_{mk})$, \;\;\;$f_{m}(\mathbf{P}_{m}) = \sum_{k=1}^{K}  q_{m} \frac{p_{mk} \lambda_k P_k^{\delta}}{W(\beta_k)\sum_{i=1}^{K}p_{mi} \lambda_i P_i^{\delta}+V(\beta_k) \sum_{i=1}^{K}\lambda_i P_i^{\delta}}$, \\$h_m (\mathbf{P}_m) = [p_{m1},\cdots, p_{mk}]^{T}$ and $\mathbf{C}=[C_1,\cdots,C_K]^T$. 
Comparing Eq. (\ref{Opt_Transfer}) and (4) in \cite{TimeSharing}, we find that they have the same structure.

(2) Prove Problem (P0) satisfies the time-sharing condition introduced in \cite{TimeSharing}:

Let $\mathbf{P}_{mx}^{*}$ and $\mathbf{P}_{my}^{*}$ be optimal solutions to optimization problem (P0) with the constraint $\mathbf{C} = \mathbf{C}_{x}$ and $\mathbf{C} = \mathbf{C}_{y}$, respectively. To prove Problem (P0) satisfies the time-sharing condition introduced in \cite{TimeSharing}, we need to construct a feasible cache placement strategy $\mathbf{P}_{mz}$, such that the hit probability is at least $\upsilon \sum_{m=1}^{M} f_{m} (\mathbf{P}_{mx}^{*}) + (1-\upsilon) \sum_{m=1}^{M} f_{m} (\mathbf{P}_{my}^{*})$ with cache capacity at most $\upsilon \mathbf{C}_{x} + (1-\upsilon) \mathbf{C}_{y}$ for all $\upsilon$ between zero and one.

In our system, each BS corresponds to a cache placement strategy. Thus, such a feasible cache placement strategy  $\mathbf{P}_{mz}$ can be constructed by dividing the whole plane into two parts: the BSs in $\upsilon$ portion of which have $\mathbf{P}_{mz}=\mathbf{P}_{mx}^{*}$ and the BSs in $(1-\upsilon)$ portion of which have $\mathbf{P}_{mz}=\mathbf{P}_{my}^{*}$. Obviously, the resulting $\mathbf{P}_{mz}$ satisfies the cache capacity constraint $\upsilon \mathbf{C}_{x} + (1-\upsilon) \mathbf{C}_{y}$ and its hit probability achieves $\upsilon \sum_{m=1}^{M} f_{m} (\mathbf{P}_{mx}^{*}) + (1-\upsilon) \sum_{m=1}^{M} f_{m} (\mathbf{P}_{my}^{*})$. Therefore, the cache placement optimization problem satisfies the time-sharing condition.

\small
\bibliographystyle{IEEEtran}

\begin{IEEEbiography}
[{\includegraphics[width=1in,clip,keepaspectratio]{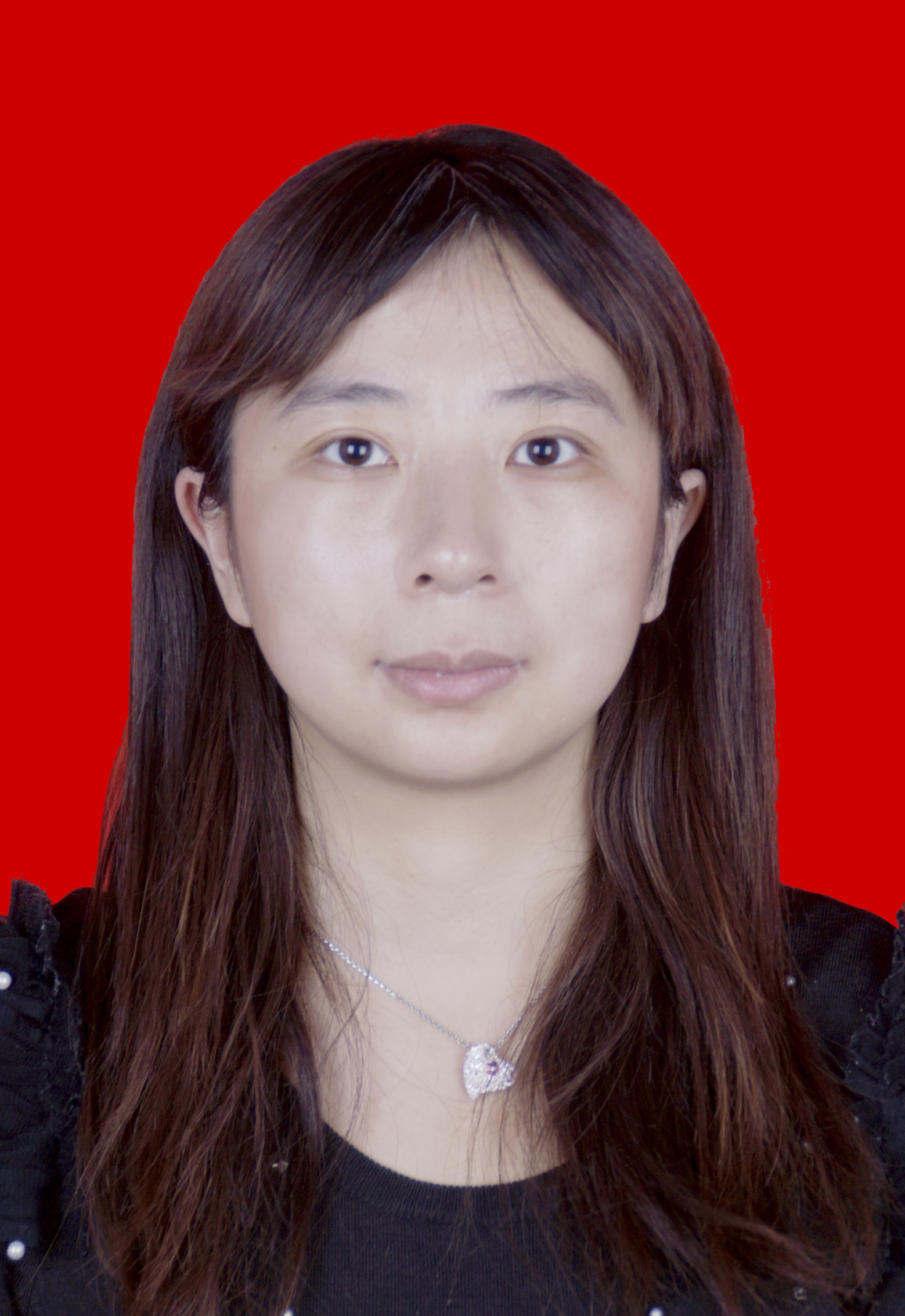}}]
{Juan Wen} received the B.Eng and the Ph. D. degree in telecommunications engineering at Xidian University, Xi'an, Shaanxi, China. Since Apr. 2016, she has been a postdoctoral researcher in the Dept. of Electrical and Electronics Engineering (EEE) at the University of Hong Kong. She was a visiting student at the University of Toronto from Sep. 2013 to Aug. 2014. Her research interests focus on the analysis and design of wireless networks using stochastic geometry. She was a recipient of the Best Paper Award at IEEE/CIC ICCC 2013.
\end{IEEEbiography}

\begin{IEEEbiography}
[{\includegraphics[width=1in,clip,keepaspectratio]{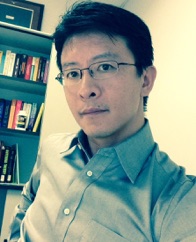}}]
{Kaibin Huang} (M'08-SM'13) received the B.Eng. (first-class hons.) and the M.Eng. from the National University of Singapore, respectively, and the Ph.D. degree from The University of Texas at Austin (UT Austin), all in electrical engineering.

Since Jan. 2014, he has been an assistant professor in the Dept. of Electrical and Electronic Engineering (EEE) at The University of Hong Kong. He used to be a faculty member in the Dept. of Applied Mathematics (AMA) at the Hong Kong Polytechnic University (PolyU) and the Dept. of EEE at Yonsei University. His research interests focus on the analysis and design of wireless networks using stochastic geometry and multi-antenna techniques.

He frequently serves on the technical program committees of major IEEE conferences in wireless communications. Most recently, he served as the lead chairs for the Wireless Comm. Symp. of IEEE Globecom 2017 and the Comm. Theory Symp. of IEEE GLOBECOM 2014 and the TPC Co-chairs for IEEE PIMRC 2017 and the IEEE CTW 2013. Currently, he is an editor for IEEE Transactions on Green Communications and Networking, and IEEE Transactions on Wireless Communications. He was an editor for  IEEE Journal on Selected Areas in Communications (JSAC) series on Green Communications and Networking in 2015-2016, for IEEE Wireless Communications Letters in 2011-2016, and for IEEE/KICS Journal of Communication and Networks in 2009-2015. He has edited a JSAC 2015 special issue on communications powered by energy harvesting.  Dr. Huang received the 2015 IEEE ComSoc Asia Pacific Outstanding Paper Award, Outstanding Teaching Award from Yonsei, Motorola Partnerships in Research Grant, the University Continuing Fellowship from UT Austin, and a Best Paper Award from IEEE GLOBECOM 2006 and PolyU AMA in 2013. 
\end{IEEEbiography}

\begin{IEEEbiography}
[{\includegraphics[width=1in,clip,keepaspectratio]{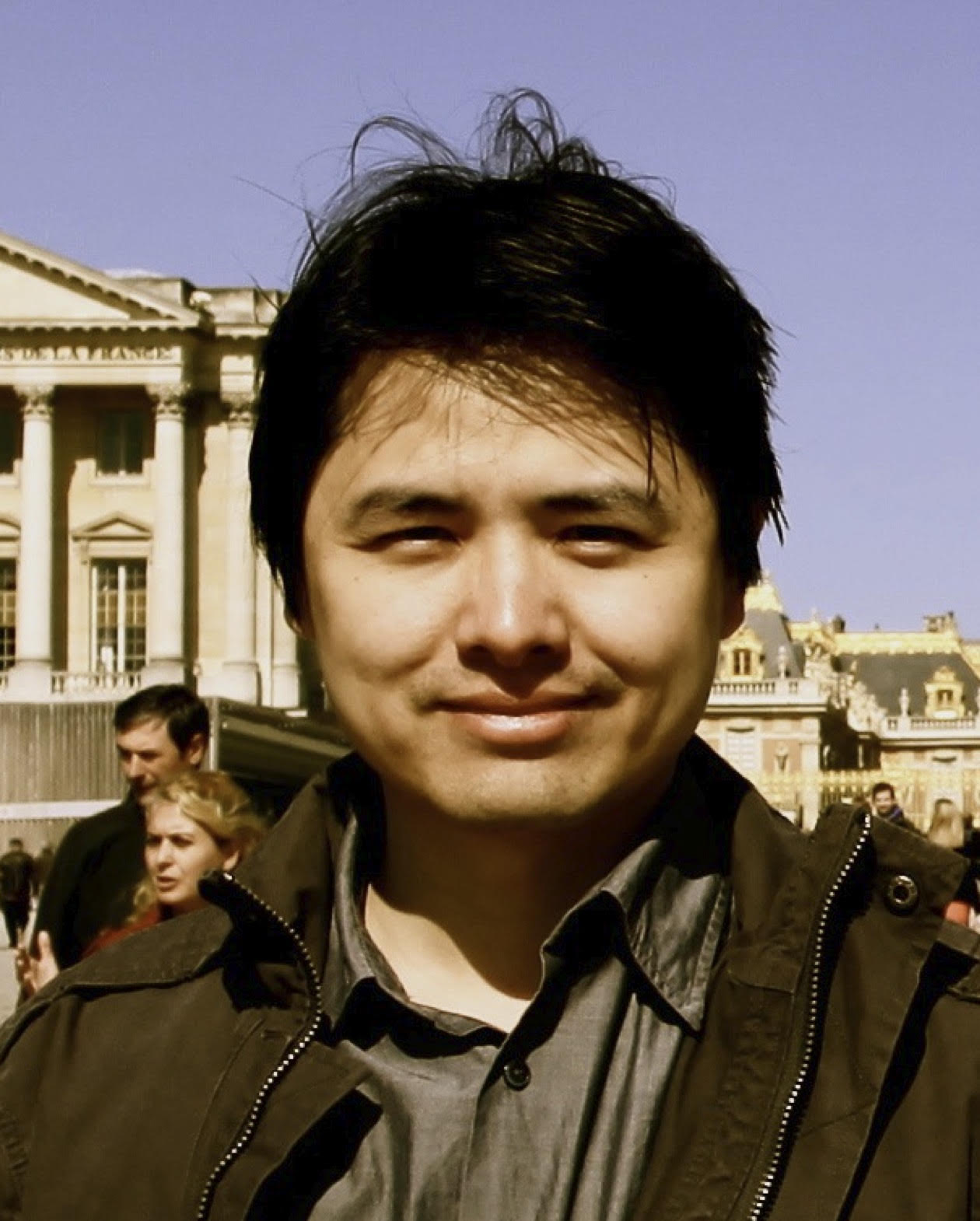}}]
{Sheng Yang} (M'07) received the B.E. degree in electrical engineering from Jiaotong University, Shanghai, China, in 2001, and both the engineer degree and the M.Sc. degree in electrical engineering from Telecom ParisTech, Paris, France, in 2004, respectively. In 2007, he obtained his Ph.D. from Université de Pierre et Marie Curie (Paris VI). From October 2007 to November 2008, he was with Motorola Research Center in Gif-sur-Yvette, France, as a senior staff research engineer. Since December 2008, he has joined CentraleSupélec where he is currently an associate professor. From April 2015, he also holds an honorary faculty position in the department of electrical and electronic engineering of the University of Hong Kong (HKU). He received the 2015 IEEE ComSoc Young Researcher Award for the Europe, Middle East, and Africa Region (EMEA). He is an editor of the IEEE transactions on wireless communications.
\end{IEEEbiography}

\begin{IEEEbiography}
[{\includegraphics[width=1in,clip,keepaspectratio]{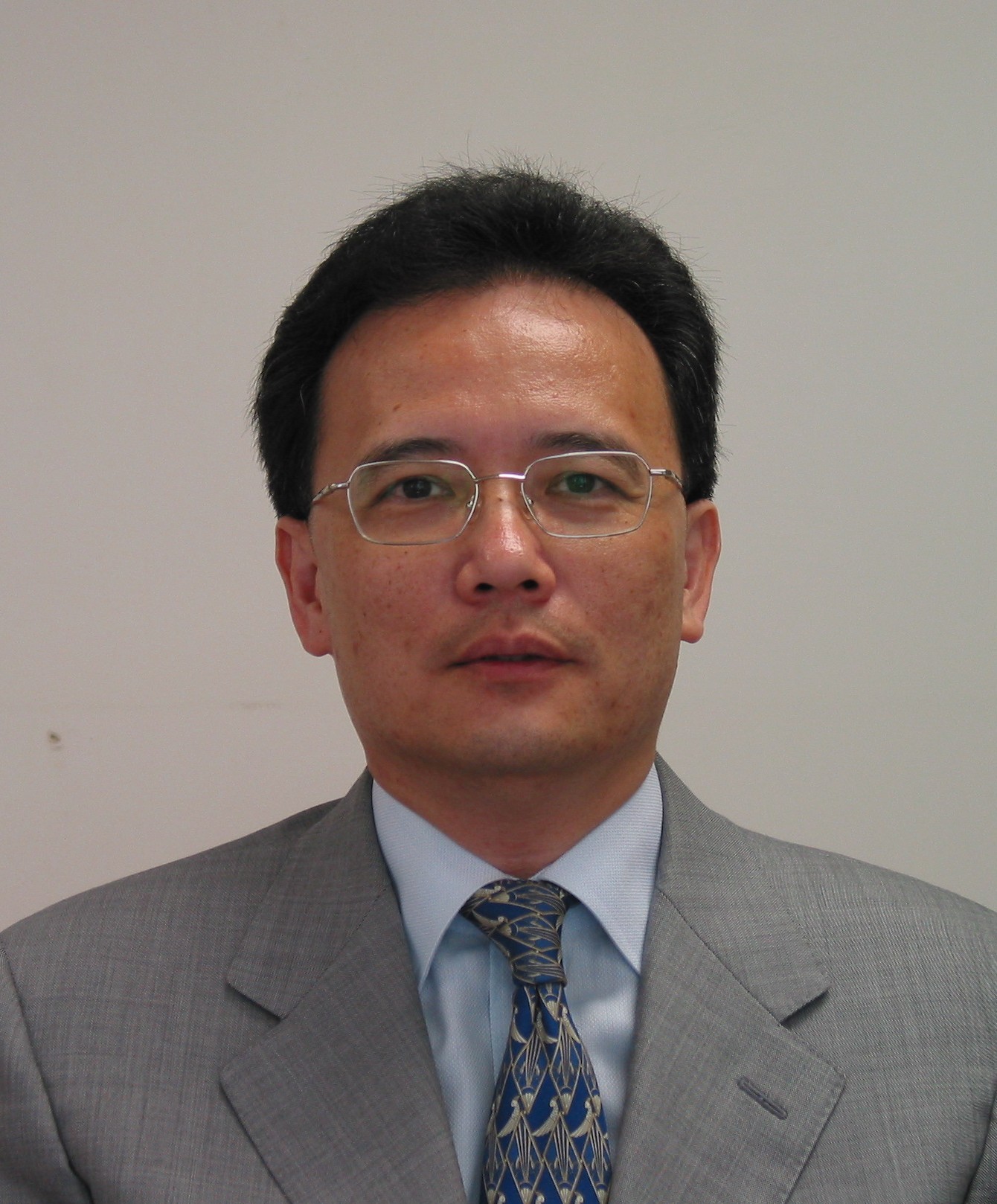}}]
{Victor O.K. Li} (S'80-M'81-F'92) received SB, SM, EE and ScD degrees in Electrical Engineering and Computer Science from MIT in 1977, 1979, 1980, and 1981, respectively. He is Chair Professor of Information Engineering, Cheng Yu-Tung Professor in Sustainable Development,  and Head of the Department of Electrical and Electronic Engineering at the University of Hong Kong (HKU).  He has also served as Assoc. Dean of Engineering and Managing Director of Versitech Ltd., the technology transfer and commercial arm of HKU.  He served on the board of China.com Ltd., and now serves on the board of Sunevision Holdings Ltd., listed on the Hong Kong Stock Exchange.   Previously, he was Professor of Electrical Engineering at the University of Southern California (USC), Los Angeles, California, USA, and Director of the USC Communication Sciences Institute. His research is in the technologies and applications of information technology, including clean energy and environment, social networks, wireless networks, and optimization techniques.  Sought by government, industry, and academic organizations, he has lectured and consulted extensively around the world.   He has received numerous awards, including the PRC Ministry of Education Changjiang Chair Professorship at Tsinghua University, the UK Royal Academy of Engineering Senior Visiting Fellowship in Communications, the Croucher Foundation Senior Research Fellowship, and the Order of the Bronze Bauhinia Star, Government of the Hong Kong Special Administrative Region, China.  He is a Registered Professional Engineer and a Fellow of the Hong Kong Academy of Engineering Sciences, the IEEE, the IAE, and the HKIE.
\end{IEEEbiography}

\end{document}